\documentclass[11pt]{article}

\usepackage[sort,numbers]{natbib}
\usepackage[bookmarks=true,bookmarksnumbered=true,colorlinks=true,allcolors=black]{hyperref}

\usepackage{graphicx}
\usepackage{amsmath,amssymb,amsfonts,amsthm}
\usepackage{array}   
\newcolumntype{C}{>{$}c<{$}} 
\usepackage{lscape}
\usepackage{arydshln}
\renewcommand*{\baselinestretch}{1.25}

\usepackage{geometry}
\usepackage{indentfirst}

\usepackage{enumitem}
\usepackage{color}

\newtheorem{theorem}{Theorem}
\newtheorem{lemma}{Lemma}
\newtheorem{proposition}{Proposition}

\newtheorem{lemA}{Lemma}[section]

\newtheorem{coroA}{Corollary}[section]
\theoremstyle{definition}

\newtheorem*{rmk*}{Remark}
\newtheorem{rmk}{Remark}
\newtheorem{assumption}{Assumption}

\renewcommand*\proofname{\upshape{\bfseries{Proof}}}

\allowdisplaybreaks

\makeatletter
    \renewcommand*{\section}{\@startsection{section}{1}{\z@}%
    {10pt}{5pt}{\reset@font\normalsize\bfseries}}
\makeatother
    
\makeatletter
    \renewcommand*{\subsection}{\@startsection{subsection}{2}{\z@}%
    {5pt}{5pt}{\reset@font\normalsize\mdseries\itshape}}
\makeatother

\makeatletter
    \renewcommand*{\subsubsection}{\@startsection{subsubsection}{3}{\z@}%
    {5pt}{5pt}{\reset@font\normalsize\mdseries\itshape}}
\makeatother

\makeatletter
\def\@listi{\leftmargin\leftmargini
  \topsep=.5\baselineskip 
  \partopsep=0pt \parsep=0pt \itemsep=0pt}
\let\@listI\@listi
\@listi
\def\@listii{\leftmargin\leftmarginii
  \labelwidth\leftmarginii \advance\labelwidth-\labelsep
  \topsep=0pt \partopsep=0pt \parsep=0pt \itemsep=0pt}
\def\@listiii{\leftmargin\leftmarginiii
  \labelwidth\leftmarginiii \advance\labelwidth-\labelsep
  \topsep=0pt \partopsep=0pt \parsep=0pt \itemsep=0pt}
\def\@listiv{\leftmargin\leftmarginiv
  \labelwidth\leftmarginiv \advance\labelwidth-\labelsep
  \topsep=0pt \partopsep=0pt \parsep=0pt \itemsep=0pt}
\makeatother

\makeatletter
\renewenvironment{proof}[1][\proofname]{\par
  \pushQED{\qed}%
  \normalfont \topsep6\p@\@plus6\p@\relax
  \trivlist
  \item[\hskip\labelsep
        \bfseries
    #1\@addpunct{.}]\ignorespaces
}{%
  \popQED\endtrivlist\@endpefalse
}
\makeatother

\newcommand{\tcr}[1]{\textcolor{black}{#1}}

\newcommand{\ol}[1]{\overline{#1}}
\newcommand{\ul}[1]{\underline{#1}}
\newcommand{\pck}[1]{{\bfseries#1}}

\title{Multi-scale analysis of lead-lag relationships in high-frequency financial markets}
\author{Takaki Hayashi\thanks{Graduate School of Business Administration, Keio University, 4-1-1 Hiyoshi, Yokohama 223-8526, Japan}
\thanks{Department of Business Administration, Graduate School of Social Sciences, Tokyo Metropolitan University, Marunouchi Eiraku Bldg. 18F, 1-4-1 Marunouchi, Chiyoda-ku, Tokyo 100-0005 Japan}
\thanks{CREST, Japan Science and Technology Agency}
\and
Yuta Koike
\thanks{Mathematics and Informatics Center and Graduate School of Mathematical Sciences, The University of Tokyo, 3-8-1 Komaba, Meguro-ku, Tokyo 153-8914 Japan}
\footnotemark[2]
\thanks{The Institute of Statistical Mathematics, 10-3 Midori-cho, Tachikawa, Tokyo 190-8562, Japan}
\footnotemark[3]}

\begin{document}

\maketitle

\begin{abstract}

We propose a novel estimation procedure for scale-by-scale lead-lag relationships of financial assets observed at high-frequency in a non-synchronous manner. The proposed estimation procedure does not require any interpolation processing of original datasets and is applicable to those with highest time resolution available. Consistency of the proposed estimators is shown under the continuous-time framework that has been developed in our previous work \cite{HK2016}. 
An empirical application to a quote dataset of the NASDAQ-100 assets identifies two types of lead-lag relationships at different time scales. 
\vspace{3mm}


\end{abstract}

\section{Introduction}

A financial market accommodates a diversified groups of participants. They have different sources of money, different time horizons and different risk attitudes, with different quality and quantity of information. In \citet{MDDPOW1993} it is argued that such differences are engraved in price formation at each of distinct time scales. They can cause a {\it multi-scale} structure embedded in the financial market.


This paper intends to study such a multi-scale structure of financial markets that can exist in a very short time period. In particular, we are to investigate {\it lead-lag relationships} between financial assets by the use of high-frequency data. Identification of lead-lag relationships among assets is fundamentally important both for theoretical and practical perspectives; the existence of such relationships may mean the inefficiency of financial markets for theorists but it may also provide opportunities for market participants to earn ``excess'' profits. So that so, it is quite natural that lead-lag analysis has been conducted in the finance literature for a long time. Since 90's as high-frequency data has become more and more accessible, lead-lag relationships with high-frequency data have been studied by such authors as ~\cite{Chan1992,deJN1997,Reno2003,HA2014}. 
In the meantime, multi-scale analysis with high-frequency financial data has been carried out; e.g., \cite{MDDOPW1997,Hasbrouck2015,BV2015,Subbotin2008,GGSB2010}. However, main interest of most of these articles is the estimation of volatilities of assets. There is little work that conducts multi-scale analysis of lead-lag relationships in the high-frequency domain; one exception is \citet{Hafner2012} which has examined multi-scale structures of the lead-lag relationships between the returns, durations and volumes of high-frequency transaction data of the IBM stock. 
%

To our understanding, the main focus of those studies conducting multi-scale analysis is empirical application {\it per se}, not to develop a new estimation methodology. Their adopted approaches are theoretically based on ``classical'' discrete time series that appear to be more suitable for daily or lower frequency data with longer time horizons. 
On one hand, analysis of high-frequency financial data shall focus on a short time horizon, that is, one day or shorter. So, it is unclear whether one can reasonably apply such a ``classical'' method to high-frequency financial data without reservation. 
On the other hand, {\it continuous-time} modeling provides a convenient and powerful framework to analyze high-frequency data observed in a short horizon (cf.~\citet{AJ2014}). 

With these in mind, in \cite{HK2016} the authors have developed a continuous-time framework that is designed specifically for multi-scale analysis of lead-lag relationships in high-frequency data. 
There, they introduce two Brownian motions $B^1$ and $B^2$ with a scale-by-scale correlation structure. 	
More precisely, they have shown that, for any $R_j\in[-1,1]$ and $\theta_j\in\mathbb{R}$ ($j=0,1,\dots$), there exists a bivariate Gaussian process $B_t=(B^1_t,B^2_t)$ ($t\in\mathbb{R}$) with stationary increments such that
\begin{enumerate}[label=(\Roman*)]

\item\label{bm} both $B^1$ and $B^2$ are two-sided Brownian motions,

\item\label{cross-spec} the cross-spectral density of $B$ is given by
\begin{equation}\label{csd}
f(\lambda)=\sum_{j=0}^\infty R_je^{-\sqrt{-1}\theta_j\lambda}1_{\Lambda_j}(\lambda),\qquad\lambda\in\mathbb{R},
\end{equation}
where $\Lambda_j=[-2^{j+1}\pi,-2^{j}\pi)\cup(2^{j}\pi,2^{j+1}\pi]$ for every $j\in\mathbb{Z}$.

\end{enumerate}
The frequency band $\Lambda_j$ corresponds to the time scale between $2^{-j}$ and $2^{-j+1}$ in the time domain. Also, note that, if $W_t=(W^1_t,W^2_t)$ ($t\in\mathbb{R}$) is a two-sided bivariate Brownian motion with correlation $R$, for $\theta\in\mathbb{R}$ the process $(W^1_t,W^2_{t-\theta})$ ($t\in\mathbb{R}$) has the cross-spectral density $Re^{-\sqrt{-1}\theta\lambda}$ ($\lambda\in\mathbb{R}$). Therefore, we can consider that $B^1$ and $B^2$ have a lead-lag relationship with the time-lag $\theta_j$ in the time scale between $2^{-j}$ and $2^{-j+1}$. Hence, under this model we can understand the multi-scale structure of the lead-lag relationships by estimating the parameters $\theta_j$ from observation data.

The main contribution of this paper is to develop a novel estimation procedure for the parameters $\theta_j$ based on {\it non-synchronous} observations of (volatility-modulated versions of) $B^1$ and $B^2$. 
In the above mentioned \cite{HK2016} the authors proposed another estimation procedure, which required data interpolation in accordance with a regular grid with size equated to the finest time resolution at which the fastest market participants (will) act. When analyzing a dataset with sub mili-second time precision, one typically wishes to let this finest resolution coarser than the actual time precision (see Section \ref{section:empirical} for instance). If so, such an \textit{intermediary} data interpolation step can inevitably discard a large amount of data. 
Even in such a situation the newly proposed procedure in this paper is free from any interpolation processing of the original data and able to efficiently use them. 
Besides, a theoretical consideration along with numerical experiments suggests that the new estimators can potentially have better performance than the interpolation-based estimator when the sampling times are non-synchronous to a reasonable degree. 
An empirical application with a NASDAQ-100 dataset identifies two types of lead-lag relationships at different time scales. At the best of our knowledge, this observation is new in the empirical literature, indicating potential usefulness of the new estimation methodology.

The rest of the paper is organized as follows. In Section \ref{section:setting} we present the theoretical setting considered in this paper in details. Our new estimation procedure is described in Section \ref{section:estimation}. We develop an asymptotic theory associated with the proposed estimators in Section \ref{section:asymptotic}.
In Section \ref{section:simulation} we assess finite sample performance of the proposed estimators by Monte Carlo experiments, and in Section \ref{section:empirical} we apply our procedure to empirical datasets. Section \ref{section:conclusion} concludes the paper. All the proofs are collected in Section \ref{section:proof}.

\section{Setting}\label{section:setting}

We let the finest time resolution correspond to $\tau_N:=2^{-N-1}$ for some $N\in\mathbb{N}$. We suppose $\tau_N$ is comparable to the observation frequency of data. We will develop an asymptotic theory in the high-frequency setting, i.e.,  when $N$ tends to infinity, or the time resolution shrinks to zero, while the length of the whole observation interval stays fixed.

As mentioned in the Introduction, our theoretical framework is based on a bivariate Gaussian process $B_t=(B^1_t,B^2_t)$ ($t\in\mathbb{R}$) with stationary increments satisfying properties \ref{bm}--\ref{cross-spec}. Since we are mainly interested in the lead-lag relationships at scales close to the finest time resolution, it is convenient to ``relabel'' indices of the parameters $R_j$ and $\theta_j$ in \eqref{csd} so that the finest resolution $\tau_N$ corresponds to the level $j=1$ while we consider the asymptotic theory such that $N$ tends to infinity. For this reason, as in \cite{HK2016} we replace property \ref{cross-spec} with the following one: The cross-spectral density of $B$ is given by
\begin{equation}\label{asymptotic}
f_N(\lambda)=\sum_{j=1}^{N+1}R_je^{-\sqrt{-1}\theta_j\lambda}1_{\Lambda_{\tcr{N-j+1}}}(\lambda),\qquad\lambda\in\mathbb{R}.
\end{equation}
We also assume that $\theta_j\in(-\delta,\delta)$ for every $j$ with some $\delta>0$. 

Now, for each $\nu=1,2$, we consider the log price process $X^\nu=(X^\nu_t)_{t\geq0}$ of the $\nu$-th asset given by
\begin{equation}\label{model}
X^\nu_t=X^\nu_0+\int_0^t\sigma^\nu_sdB^\nu_s,\qquad t\geq0,
\end{equation}
where $(\sigma^\nu_t)_{t\geq0}$ is a c\`adl\`ag process adapted to the filtration $(\mathcal{F}^\nu_t)$ such that the process $(B^\nu_t)$ is, respectively, a one-dimensional $(\mathcal{F}^\nu_t)$-Brownian motion. We observe the process $X^\nu$ on the interval $[0,T+\delta]$ at the sampling times $0\leq t^\nu_0<t^\nu_1<\cdots<t^\nu_{n_\nu}\leq T+\delta$. The sampling times $(t^1_i)_{i=0}^{n_1}$ and $(t^2_i)_{i=0}^{n_2}$ are random variables which are independent of $(X^1,X^2)$ and implicitly depend on $N$ such that
\[
r_N:=\max_{\nu=1,2}\max_{i=0,1,\dots,n_\nu+1}(t^\nu_i-t^\nu_{i-1})\to^p0
\]
as $N\to\infty$, where we set $t^\nu_{-1}:=0$ and $t^\nu_{n_\nu+1}:=T+\delta$ for each $\nu=1,2$.

\begin{rmk}
Our model is generally not a semimartingale, so it is generally not free of arbitrage in the absence of market frictions due to the well-known fundamental theorem of asset pricing (see e.g.~\cite{DS1994}). 
However, if we take account of market frictions such as discrete trading or transaction costs, we can show that our model has no arbitrage; see \cite{HK2017arb} for details.  
\end{rmk}

\section{Construction of the estimators}\label{section:estimation}

Our aim is to estimate the parameters $\theta_j$ for each $j$ based on discrete observation data $(X^1_{t_i^1})_{i=0}^{n_1}$ and $(X^2_{t_i^2})_{i=0}^{n_2}$. We begin by introducing some notation. For each $\nu=1,2$, we associate the observation times $(t^\nu_i)_{i=0}^{n_\nu}$ with the collection of intervals $\Pi^\nu_N=\{(t^\nu_{i-1},t^\nu_i]:i=1,\dots,n_\nu\}$. We will systematically employ the notation $I$ (resp.~$J$) for an element of $\Pi^1_N$ (resp.~$\Pi^2_N$). 

For an interval $H\subset[0,\infty)$, we set $\overline{H}=\sup H$, $\underline{H}=\inf H$, $|H|=\overline{H}-\underline{H}$. In addition, we set $V(H)=V_{\overline{H}}-V_{\underline{H}}$ for a a stochastic process $(V_t)_{t\geq0}$, and $H_\theta=H+\theta$ for a real number $\theta$.

Now we explain how to construct our estimators. To explain the idea behind the construction, we focus on the case of $\sigma^\nu_s\equiv1$ for $\nu=1,2$. The parameter $\theta_j$ is the unique maximizer of the scale-by-scale cross-covariance function $\rho_{\tcr{N-j+1}}(\theta)$ between $B^1$ and $B^2$, which is defined by
\[
\rho_{\tcr{N-j+1}}(\theta)=E\left[\left(\int_{-\infty}^\infty\psi^{LP}_{\tcr{N-j+1}}(s-u)dB^1_s\right)\left(\int_{-\infty}^\infty\psi^{LP}_{\tcr{N-j+1}}(s-u-\theta)dB^2_s\right)\right],\qquad\theta\in\mathbb{R},
\]
where $\psi^{LP}_{\tcr{k}}(s)=2^{\tcr{k}/2}\psi^{LP}(2^{\tcr{k}}s)$ for $k\in\mathbb{Z}$ and $\psi^{LP}$ denotes the Littlewood-Paley wavelet:
\[
\psi^{LP}(s)=(\pi s)^{-1}(\sin(2\pi s)-\sin(\pi s))
\]
(see Sections 2.2--2.3 of \cite{HK2016} for details). Motivated by this fact, we first construct a sensible {\it covariance} estimator $\widehat{\rho}_{\tcr{N-j+1}}(\theta)$ for $\rho_{\tcr{N-j+1}}(\theta)$, and then construct the {\it lead-lag} estimator $\widehat{\theta}_{j}$ for $\theta_j$ as a maximizer of $|\widehat{\rho}_{\tcr{N-j+1}}(\theta)|$ as in \cite{HRY2013}.
The idea behind the construction of the estimator $\widehat{\rho}_{\tcr{N-j+1}}(\theta)$ is as follows. Let $U^N(\theta)$ be the inverse Fourier transform of $f_N(\lambda)$. Then we have
\[
\rho_{\tcr{N-j+1}}(\theta)=2^{-\frac{\tcr{N-j+1}}{2}}(U^N*\psi^{LP}_{\tcr{N-j+1}})(\theta)=\int_{-\infty}^\infty U^N(\theta-s)\psi^{LP}(2^{\tcr{N-j+1}}s)ds
\]
by the convolution theorem. This suggests us to consider the following estimator for $\rho_{\tcr{N-j+1}}(\theta)$:
\[
\widehat{\rho}_{\tcr{N-j+1}}(\theta):=\widehat{\rho}_{\tcr{N-j+1}}^N(\theta)=\sum_{l=-L_j+1}^{L_j-1}\widehat{U}^N(\theta-l\tau_N)\Psi_{j}(l),	
\] 
where $\widehat{U}^N(\theta)$ is an estimator for $U^N(\theta)$ and $\Psi_{j}(l)$ is an approximation of $\psi^{LP}(2^{\tcr{N-j+1}}l\tau_N)$ (it turns out that the factor $\tau_N$ corresponding to $ds$ is unnecessary because $2^{\tcr{N-j+1}}\tau_N=2^{-j}$ does not tend to 0 in our asymptotic setting), 
both of 	
which are explicitly defined in the following. Since $U^N(\theta)$ may be regarded as the ``cross-covariance function between $dB^1$ and $dB^2$'', we adopt the following estimator introduced in \citet{HRY2013} as $\widehat{U}^N(\theta)$: 
\[
\widehat{U}^N(\theta)=
\left\{\begin{array}{ll}
\sum_{I\in\Pi^1_N,J\in\Pi^2_N:\overline{I}\leq T}X^1(I)X^2(J)K(I,J_{-\theta})&\text{if }\theta\geq0,\\
\sum_{I\in\Pi^1_N,J\in\Pi^2_N:\overline{J}\leq T}X^1(I)X^2(J)K(I_\theta,J)&\text{if }\theta<0,
\end{array}\right.
\]
where we set $K(I,J)=1_{\{I\cap J\neq\emptyset\}}$ for two intervals $I$ and $J$. This $\widehat{U}^N(\theta)$ can be regarded as the empirical cross-covariance estimator between the returns of $X^1$ and $X^2$ at the lag $\theta$ computed by \citet{HY2005}'s method to handle the non-synchronous sampling times. In the meantime, the  Fourier inversion formula yields
\[
\psi^{LP}(2^{\tcr{N-j+1}}l\tau_N)=\frac{2^j\tau_N}{2\pi}\int_{-\infty}^\infty e^{\sqrt{-1}l\tau_N\lambda}1_{\Lambda_{N-j+1}}(\lambda)d\lambda
=\frac{1}{2\pi}\int_{-\pi}^\pi e^{\sqrt{-1}l\lambda}\cdot2^j1_{\Lambda_{-j}}(\lambda)d\lambda,
\]
so the transfer function of $(\psi^{LP}(2^{\tcr{N-j+1}}l\tau_N))_{l\in\mathbb{Z}}$ is $2^j1_{\Lambda_{-j}}(\lambda)$. In particular, $\Psi_{j}(l)$ well approximates $\psi^{LP}(2^{\tcr{N-j+1}}l\tau_N)$ if the transfer function of $(\Psi_j(l))_{l=-L_j+1}^{L_j-1}$ well approximates $2^j1_{\Lambda_{-j}}(\lambda)$. We construct such a sequence $(\Psi_j(l))_{l=-L_j+1}^{L_j-1}$ from Daubechies' wavelet filter as follows. We refer to Section 4.8 of \cite{PW2000} for details about Daubechies' wavelets (see also Appendix \ref{sec:appendix}). 
Let $(h_p)_{p=0}^{L-1}$ be Daubechies' wavelet filter of (even) length $L$ whose power transfer function $H_L(\lambda)=|\sum_{p=0}^{L-1}h_{p}e^{-\sqrt{-1}\lambda p}|^2$ is given by
\[
H_L(\lambda)=2\sin^L(\lambda/2)\sum_{p=0}^{L/2-1}\binom{L/2-1+p}{p}\cos^{2p}(\lambda/2),\qquad\lambda\in\mathbb{R}.
\]
The associated scaling filter\footnote{We use the notation that $(h_p)$ denotes the wavelet filter and $(g_p)$ denotes the scaling filter following \cite{PW2000}. Note that the reverse notation is often used in the literature.} $(g_p)_{p=0}^{L-1}$ is defined via the quadrature mirror relationship as $g_p=(-1)^{p+1}h_{L-p-1}$, $p=0,1,\dots,L-1$, hence its power transfer function $G_L(\lambda)=|\sum_{p=0}^{L-1}g_{p}e^{-\sqrt{-1}\lambda p}|^2$ satisfies $G_L(\lambda)=H_L(\lambda-\pi)$. Then, for every $j$ we construct the associated level $j$ wavelet filter $(h_{j,p})_{p=0}^{L_j-1}$ recursively by $h_{1,p}=h_p$ for $p=0,1,\dots,L_1-1$ and $h_{j,p}=\sum_{q=0}^{L_{j-1}-1}g_{p-2q}h_{j-1,q}$ for $p=0,1,\dots,L_j-1$, where $L_j=(2^j-1)(L-1)+1$ and $g_p=0$ for $p\notin\{0,1,\dots,L-1\}$. Now we define the sequence $(\Psi_j(l))_{l=-L_j+1}^{L_j-1}$ by
\[
\Psi_j(l)=\sum_{p=0}^{L_j-1-|l|}h_{j,p}h_{j,p+|l|},\qquad l=0,\pm1,\dots,\pm(L_j-1).
\]
These quantities are identical to the \textit{autocorrelation wavelets} from \citet{NvSK2000} (see Definition 3 from \cite{NvSK2000}). The transfer function $H_{j,L}(\lambda)=\sum_{l=-L_j+1}^{L_j-1}\Psi_j(l)e^{-\sqrt{-1}l\lambda}$ of $(\Psi_j(l))_{l=-L_j+1}^{L_j-1}$ is given by
\begin{equation*}
H_{j,L}(\lambda)=H_L(2^{j-1}\lambda)\prod_{i=0}^{j-2}G_L(2^i\lambda),\qquad\lambda\in\mathbb{R}
\end{equation*}
(see Eq.(28) from \cite{NvSK2000}). 
In particular, $H_{j,L}(\lambda)$ well approximates $2^j1_{\Lambda_{-j}}(\lambda)$ as $L\to\infty$ (see \eqref{lai}) and thus $\Psi_j(l)$ may be used an approximation of $\psi^{LP}(2^{\tcr{N-j+1}}l\tau_N)$. 
Finally, for every $j\in\mathbb{N}$ we define the estimator $\widehat{\theta}_j:=\widehat{\theta}_j^N$ for $\theta_j$ as a solution of the following equation:
\[
\left|\widehat{\rho}_{\tcr{N-j+1}}(\widehat{\theta}_j)\right|=\max_{\theta\in\mathcal{G}^N}\left|\widehat{\rho}_{\tcr{N-j+1}}(\theta)\right|.
\]
Here, we maximize the function $\widehat{\rho}_{\tcr{N-j+1}}(\theta)$ regarding $\theta$ over the finite grid
\[
\mathcal{G}^N=\{l\tau_N:l\in\mathbb{Z},|l|\leq \Gamma_N\}
\]
with some positive integer $\Gamma_N$ as in \cite{HRY2013}.

\begin{rmk}
Given the length $L$ of Daubechies' wavelet filter, we still have several options of $(h_{p})_{p=0}^{L-1}$ such as the \textit{external phase wavelet} and the \textit{least asymmetric wavelet} (cf.~Section 4.8 of \cite{PW2000}). However, all of them have the same power transfer function $H_L(\lambda)$ by definition, so $(\Psi_j(l))_{l=-L_j+1}^{L_j-1}$ only depends on the length $L$ of Daubechies' wavelet filters.
\end{rmk}

\section{Asymptotic theory}\label{section:asymptotic}

For a function $f\in L^1(\mathbb{R})$, we denote by $\mathcal{F}f$ the Fourier transform of $f$:
\[
(\mathcal{F}f)(\lambda)=\int_{-\infty}^\infty f(t)e^{-\sqrt{-1}\lambda t}dt,\qquad\lambda\in\mathbb{R}.
\]
We impose the following conditions to derive our asymptotic results.
\begin{assumption}\label{volatility}
There is a constant $\gamma\in(0,1]$ such that $\sigma^\nu$ almost surely has $\gamma$-H\"older continuous sample paths for every $\nu=1,2$.
\end{assumption}

\begin{assumption}\label{sampling}
(i) $r_N=O_p(\tau_N^\xi)$ as $N\to\infty$ for any $\xi\in(0,1)$.

\noindent(ii) There are constants $\alpha>1$, $\beta\in(0,1)$, $Q>1$ and an absolutely continuous real-valued function $D$ on $[-\pi,\pi]$ such that
\begin{align*}
\tau_{m}\sum_{k=0}^{\lceil T\tau_{m}^{-1}\rceil-1}\int_{-\pi}^\pi E\left[\left|D^N_k(\lambda,\theta_N)-D(\lambda)\right|^Q\right]d\lambda
=O(\tau_{N}^{\alpha})
\end{align*}
as $N\to\infty$ for any sequence $(\theta_N)$ of real numbers satisfying $\theta_N\in\mathcal{G}^N$ for every $N$, where $m=\lceil \beta N\rceil$,
\[
D^N_k(\lambda,\theta)=
\left\{\begin{array}{ll}
\frac{1}{2\pi\tau_{m}\tau_N}\sum_{I,J:\underline{I}\in I_m(k)}(\mathcal{F}1_{I})(\lambda/\tau_N)(\mathcal{F}1_{J_{-\theta}})(-\lambda/\tau_N)K(I,J_{-\theta}) & \text{if }\theta\geq0,\\
\frac{1}{2\pi\tau_{m}\tau_N}\sum_{I,J:\underline{J}\in I_m(k)}(\mathcal{F}1_{I_\theta})(\lambda/\tau_N)(\mathcal{F}1_{J})(-\lambda/\tau_N)K(I_\theta,J) & \text{if }\theta<0
\end{array}\right.
\]
and $I_m(k)=[k\tau_m,(k+1)\tau_m)$. Moreover, $D(\lambda)>0$ for almost all $\lambda\in[-\pi,\pi]$ and $D,D'\in L^\infty(-\pi,\pi)$.

\end{assumption}

The simplest situation where Assumption \ref{sampling} is satisfied is the equidistant and synchronous sampling case such that $t^1_i=t^2_i=i\tau_N/a$ for every $i$ with some $a\in\mathbb{N}$. In this case one can easily see that  
\[
D^N_k(\lambda,\theta)=\frac{a}{2\pi }\left|\frac{e^{-\sqrt{-1}\lambda/a}-1}{\lambda}\right|^2
\]
for any $\theta\in\mathcal{G}^N$. Hence, Assumption \ref{sampling} is satisfied with $D(\lambda)$ being the quantity in the right side of the above equation. Another example is \citet{LM1990}'s sampling scheme as described by the following proposition:
\begin{proposition}\label{prop:bernoulli}
Let $a\in\mathbb{N}$. 
Suppose that, for each $\nu=1,2$, the observation times $(t^\nu_i)_{i=0}^{n_\nu}$ are randomly chosen from $\{i\tau_N/a:i=0,1,\dots,\lfloor(T+\delta)a\tau_N^{-1}\rfloor\}$ using Bernoulli trials with success probability $1-\pi_\nu$ ($0\leq\pi_\nu<1$). Then, Assumption \ref{sampling} is satisfied with
\begin{align}
D(\lambda)&=\frac{a}{\pi\lambda^2}\Re\left[\frac{(1-\pi_1)(1-\pi_2)(1-e^{-\sqrt{-1}\lambda/a})}{(1-\pi_1e^{-\sqrt{-1}\lambda/a})(1-\pi_2e^{-\sqrt{-1}\lambda/a})}\right]\nonumber\\
&=\frac{a(1-\cos(\lambda/a))}{\pi\lambda^2}\frac{(1-\pi_1)(1-\pi_2)(1+\pi_1+\pi_2-\pi_1\pi_2(2\cos(\lambda/a)+1))}{|1-\pi_1e^{-\sqrt{-1}\lambda/a}|^2|1-\pi_2e^{-\sqrt{-1}\lambda/a}|^2}.\label{eq:bernoulli}
\end{align}
\end{proposition}

Now we state asymptotic results. 
The first result concerns the asymptotic behavior of the estimators $\widehat{\rho}_{\tcr{N-j+1}}(\theta)$ and can be considered as a counterpart of Propositions 3--4 in \cite{HRY2013}:
\begin{theorem}\label{theorem:main}
Let $j$ be a positive integer. 
Suppose that Assumptions \ref{volatility}--\ref{sampling} are satisfied. 
Suppose also that $L\to\infty$ and $(L\tau_N)^\gamma\log L\to0$ as $N\to\infty$ and that $(\Gamma_N+L_j)\tau_N<\delta$ for every $N$.
\begin{enumerate}[label={\normalfont(\alph*)}]

\item If a sequence $v_N>0$ satisfies $\tau_N^{-1}v_N\to\infty$ as $N\to\infty$, then
\[
\max_{\theta\in\mathcal{G}^N:|\theta-\theta_j|\geq v_N}\left|\widehat{\rho}_{\tcr{N-j+1}}(\theta)\right|\to^p0
\]
as $N\to\infty$.

\item Let $(\vartheta_N)$ be a sequence of real numbers such that $\vartheta_N\in\mathcal{G}^N$ and $\tau_N^{-1}(\vartheta_N-\theta_j)\to b$ as $N\to\infty$ for some $b\in\mathbb{R}$. 
Then 
\[
\widehat{\rho}_{\tcr{N-j+1}}(\vartheta_N)\to^p2^j\Sigma_T(\theta_j)R_{j}\int_{\Lambda_{-j}}D(\lambda)\cos(b\lambda)d\lambda
\]
as $N\to\infty$, 
where
\[
\Sigma_T(\theta)
=\left\{
\begin{array}{ll}
\int_0^{T}\sigma^1_s\sigma^2_{s+\theta}ds  & \text{if }\theta\geq0,  \\
\int_0^{T}\sigma^1_{s-\theta}\sigma^2_{s}ds  & \text{otherwise}.
\end{array}
\right.
\]
\end{enumerate}
\end{theorem}

The next theorem concerns consistency of the estimators $\widehat{\theta}_j$ and can be considered as a counterpart of Theorem 1 in \cite{HRY2013}:
\begin{theorem}\label{HRY}
Let $j$ be a positive integer. 
Suppose that Assumptions \ref{volatility}--\ref{sampling} are satisfied. 
Suppose also that $L\to\infty$ and $(L\tau_N)^\gamma\log L\to0$ as $N\to\infty$ and that $(\Gamma_N+L_j)\tau_N<\delta$ for every $N$. 
If a sequence $v_N>0$ satisfies $\tau_N^{-1}v_N\to\infty$ as $N\to\infty$, then
\[
v_N^{-1}(\widehat{\theta}_j-\theta_j)\to^p0
\]
as $N\to\infty$, provided that $R_j\neq0$ and $\Sigma_T(\theta_j)\neq0$ a.s. In particular, we have $\widehat{\theta}_j\to^p\theta_j$ as $N\to\infty$
.
\end{theorem}
Theorem \ref{HRY} shows that the proposed estimators $\widehat{\theta}_j$ enjoy a similar asymptotic property to that of the lead-lag estimator by \cite{HRY2013}. 
Remarkably, our estimators achieve the same convergence rate as \citet{HRY2013}'s one, so we have no cost to separate multiple lead-lag relationships in our setting. 
\begin{rmk}\label{rmk:non-sy}
The convergence rate given in our result is better than the one in \cite[Theorem 2]{HK2016}, where the latter presents convergence rates of different estimators for $\theta_j$. This improvement is not due to our new construction of estimators but due to gain from a special property of Daubechies' wavelets (presented in Corollary \ref{coro:deriv}), which is ignored in \cite{HK2016}. 
\end{rmk}

\begin{rmk}
Our estimators are expected to perform better than \cite{HK2016}'s one for non-synchronous data. To see this, let us consider the Lo-MacKinlay sampling scheme considered in Proposition \ref{prop:bernoulli} with $a=1$. Then, $\widehat{\rho}_{N-j+1}(\theta_j)$ converges in probability to $\rho_j^{[1]}:=2^j\Sigma_T(\theta_j)R_{j}\int_{\Lambda_{-j}}D(\lambda)d\lambda$ as $N\to\infty$, where $D(\lambda)$ is give by \eqref{eq:bernoulli}. Meanwhile, the counterpart of $\widehat{\rho}_{N-j+1}(\theta_j)$ in \cite{HK2016} converges in probability to $\rho_j^{[2]}:=2^j\Sigma_T(\theta_j)R_{j}\int_{\Lambda_{-j}}\tilde{D}(\lambda)d\lambda$ as $N\to\infty$, where
\[
\tilde{D}(\lambda)=\frac{1}{2\pi}\left|\frac{e^{-\sqrt{-1}\lambda}-1}{\lambda}\right|^2\Re\left[\frac{(1-\pi_1)(1-\pi_2)}{(1-\pi_1e^{-\sqrt{-1}\lambda})(1-\pi_2e^{-\sqrt{-1}\lambda})}\right].
\]
See Theorem 1 in \cite{HK2016} (note that $\widehat{\rho}_{N-j+1}(\theta_j)$ is essentially constructed from observations on $[0,T+|\theta_j|]$). 
A straightforward computation shows
\begin{multline*}
|\rho_j^{[1]}|-|\rho_j^{[2]}|=(1-\pi_1)(1-\pi_2)(\pi_1+\pi_2-2\pi_1\pi_2)\\
\times2^j\Sigma_T(\theta_j)|R_{j}|\int_{\Lambda_{-j}}\frac{\sin^2\lambda}{\pi\lambda^2|1-\pi_1e^{-\sqrt{-1}\lambda}|^2|1-\pi_2e^{-\sqrt{-1}\lambda}|^2}d\lambda,
\end{multline*}
which is always non-negative and strictly positive when $\Sigma_T(\theta_j)|R_{j}|>0$ and $\pi_1\wedge\pi_2>0$. 
Therefore, we may expect that our cross-covariance estimators could identify peaks of true correlograms better when observations are non-synchronous. 
In the next section we give numerical evidence to support this argument. 
\end{rmk}

\section{Simulation study}\label{section:simulation}

In this section we assess the finite sample accuracy of the proposed estimators $\widehat{\theta}_j$ by a Monte Carlo study. We set $N=14$, $T=n\tau_N$ with $n=30,000$. 


We simulate model \eqref{model} with the following two scenarios of the volatility processes:
\begin{description}

\item[Scenario 1] Constant volatilities. $\sigma^\nu\equiv1$ for $\nu=1,2$.

\item[Scenario 2] Stochastic volatilities with a leverage effect. The Heston model is adopted to generate the volatility process $\sigma^\nu_t$ for each $\nu=1,2$: The process $v^\nu_t=(\sigma^\nu_t)^2$ is the solution of the following stochastic differential equation:
\[
dv^\nu_t=\kappa(\eta-v^\nu_t)dt+\xi\sqrt{v^\nu_t}(\rho dB^\nu_t+\sqrt{1-\rho^2}dW^\nu_t),
\]
where $W^\nu$ is a standard Wiener process and the initial value $v_0^\nu$ is randomly drawn from the stationary distribution of the process $v^\nu_t$ in each iteration, i.e.~$v_0^\nu\sim\mathrm{Gamma}(2\kappa\eta/\xi^2,2\kappa/\xi^2)$. We assume that the processes $B$, $W^1$ and $W^2$ are mutually independent. The parameters $\kappa$, $\eta$, $\xi$ and $\rho$ are chosen as in \cite{CPTV2017}: $\kappa=5$, $\eta=0.04$, $\xi=0.5$ and $\rho=-0.5$. 

\end{description}
The parameters for the spectral density \eqref{asymptotic} are chosen as in Table \ref{parameters}. Simulation of the paths of the process $B$ is performed in the same way as in \cite{HK2016}. 
Namely, we simulate the bivariate stationary Gaussian sequence $\Delta_kB:=B_{(k+1)\tau_N}-B_{k\tau_N}$ ($k=0,1,\dots,n-1$) by the multivariate circulant embedding method of \cite{CW1999} with the following approximation\footnote{Note that there are typos in the first and second displayed equations in page 1221 of \cite{HK2016}: The $j$-th summands on their right hand sides should have been multiplied by $2^j$.}:
\[
E[\Delta_kB^1\Delta_{k+l}B^2]\approx\tau_N^2\sum_{j=1}^{J+1}2^{J-j+1}R_j\psi^{LP}(2^{J-j+1}(l\tau_N-\theta_j)).
\]
This simulation scheme has been implemented in the R package \pck{yuima} as the function \texttt{simBmllag} since version 1.10.2. The R package \pck{yuima} also contains the function \texttt{wllag} to implement our scale-by-scale lead-lag time estimators $\widehat{\theta}_j$. 

\begin{table}[htp]
\caption{Parameters for the spectral density \eqref{asymptotic}}
\label{parameters}
\begin{center}
\begin{tabular}{l|*{10}{c}}\hline
$j$ & 1 & 2 & 3 & 4 & 5 & 6 & 7 & 8 & 9--15 \\ \hdashline
$R_j$ & 0.3 & 0.5 & 0.7 & 0.5 & 0.5 & 0.5 & 0.5 & 0.5 & 0 \\
$\theta_j/\tau_N$  & $-1$ & $-1$ & $-2$ & $-2$ & $-3$ & $-5$ & $-7$ & $-10$ & 0 \\ \hline
\end{tabular}
\end{center}
\end{table}%

We use the Lo-MacKinlay sampling scheme with $a=1$ presented in Section \ref{section:asymptotic} to generate the sampling times $(t^1_i)_{i=0}^{n_1}$ and $(t^2_i)_{i=0}^{n_2}$. We fix $\pi_1$ as $\pi_1=1/4$ and vary $\pi_2$ as $\pi_2\in\{1/4,1/2,3/4\}$. 
Recall that $\pi_\nu$ is the probability of occurrence of observation missing for the $\nu$-th asset $X^\nu$. The larger value $\pi_2$ takes the less frequently $X^\nu$ is observed, making the degree of non-synchronicity higher.
We use $L=20$ as the length of Daubechies' wavelet filter and set $\mathcal{G}^N=\{l\tau_N: l\in\mathbb{Z}, |l|\leq100\}.$ For comparison we also compute the estimators for $\theta_j$ proposed in \cite{HK2016}, each of which is defined as a maximizer of the corresponding so-called {\it wavelet cross-covariance estimator} based on data synchronized by interpolation (we refer to it as ``WCCF''). Here, the computation of the wavelet cross-covariance estimators requires the specification of Daubechies' wavelet and we use the least asymmetric wavelet with length 20. 
We run 1,000 Monte Carlo iterations with each of three experimental conditions in each scenario. Table \ref{table:results} reports the sample median and the median absolute deviation (MAD) of the estimates for each experiment in Scenario 1. We see from Table \ref{table:results} that both estimators accurately estimate the true values in the case of $\pi_2=1/4$ at the levels $j\leq 7$. It is theoretically natural that the accuracy of the estimators declines as $j$ increases because the contrast function $|\widehat{\rho}_{\tcr{N-j+1}}(\theta)|$, $\theta\in\mathcal{G}_N$ gets flatter as $L_j=(2^j-1)(L-1)+1$ increases. In the cases of $\pi_2=1/2$ and $\pi_2=3/4$, the WCCF estimators are apparently biased at the levels $j\geq3$, while the estimators $\widehat{\theta}_j$ still keep the good precision. Hence our new estimators can handle high-frequency data with rather a high degree of non-synchronicity, which is theoretically expected from the argument in Remark \ref{rmk:non-sy}.
Table \ref{table:heston} shows simulation results in Scenario 2. As the table reveals, the presence of a time variation and a leverage effect in the volatilities does not affect the performance of the proposed estimators, which is aligned with the obtained asymptotic theory.

\begin{table}[ht]
\caption{Simulation results in Scenario 1}
\label{table:results}
\begin{center}
\begin{tabular}{l|*{8}{C}}
  \hline
$j$ & 1 & 2 & 3 & 4 & 5 & 6 & 7 & 8 \\ 
  \hline
True & -1 & -1 & -2 & -2 & -3 & -5 & -7 & -10 \\ 
&  \multicolumn{8}{c}{$\pi_2=1/4$}  \\ 
  $\widehat{\theta}_j$ & -1~(0) & -1~(0) & -2~(0) & -2~(0) & -3~(0) & -5~(1) & -7~(3) & -9~(9) \\ 
  WCCF & -1~(0) & -1~(0) & -2~(0) & -2~(0) & -3~(0) & -5~(1) & -7~(3) & -9~(9) \\ 
  &  \multicolumn{8}{c}{$\pi_2=1/2$}  \\ 
  $\widehat{\theta}_j$ & -1~(0) & -1~(0) & -2~(0) & -2~(0) & -3~(0) & -5~(1) & -7~(3) & -9~(9) \\ 
  WCCF & -1~(0) & -1~(0) & -1~(0) & -1~(0) & -2~(0) & -4~(1) & -6~(3) & -8~(9) \\ 
  &  \multicolumn{8}{c}{$\pi_2=3/4$}  \\ 
  $\widehat{\theta}_j$ & -1~(0) & -1~(0) & -2~(0) & -2~(0) & -3~(0) & -5~(1) & -7~(3) & -9~(9) \\ 
  WCCF & -1~(0) & -1~(0) & -1~(0) & 0~(0) & 0~(0) & -2~(1) & -4~(3) & -7~(9) \\ 
   \hline
\end{tabular}\\ \vspace{2mm}

\parbox{13cm}{\small This table reports the median and the median absolute deviation (in parentheses) of the estimates in Scenario 1 (divided by $\tau_N$).} 
\end{center}
\end{table}

\begin{table}[ht]
\caption{Simulation results in Scenario 2}
\label{table:heston}
\begin{center}
\begin{tabular}{l|*{8}{C}}
  \hline
$j$ & 1 & 2 & 3 & 4 & 5 & 6 & 7 & 8 \\ 
  \hline
True & -1 & -1 & -2 & -2 & -3 & -5 & -7 & -10 \\ 
&  \multicolumn{8}{c}{$\pi_2=1/4$}  \\ 
  $\widehat{\theta}_j$ & -1~(0) & -1~(0) & -2~(0) & -2~(0) & -3~(0) & -5~(1) & -7~(3) & -9~(9) \\ 
  WCCF & -1~(0) & -1~(0) & -2~(0) & -2~(0) & -3~(0) & -5~(1) & -7~(3) & -9~(9) \\ 
  &  \multicolumn{8}{c}{$\pi_2=1/2$}  \\ 
  $\widehat{\theta}_j$ & -1~(0) & -1~(0) & -2~(0) & -2~(0) & -3~(0) & -5~(1) & -7~(3) & -9~(9) \\ 
  WCCF & -1~(0) & -1~(0) & -1~(0) & -1~(0) & -2~(0) & -4~(1) & -6~(3) & -8~(9) \\ 
  &  \multicolumn{8}{c}{$\pi_2=3/4$}  \\ 
  $\widehat{\theta}_j$ & -1~(0) & -1~(0) & -2~(0) & -2~(0) & -3~(0) & -5~(1) & -7~(3) & -9~(9) \\ 
  WCCF & -1~(0) & -1~(0) & -1~(0) & 0~(0) & 0~(0) & -2~(1) & -4~(3) & -6~(9) \\ 
   \hline
\end{tabular}\\ \vspace{2mm}

\parbox{13cm}{\small This table reports the median and the median absolute deviation (in parentheses) of the estimates in Scenario 2 (divided by $\tau_N$).} 
\end{center}
\end{table}

\clearpage

\section{Empirical application}\label{section:empirical}
In this section we apply the proposed method to actual high-frequency data in the U.S. stock market. We investigate the lead-lag relationships between quote updates of a single asset traded concurrently at multiple exchanges.\footnote{This is closely related to the issue of identifying the particular exchange at which {\it price discovery} of the assets actually takes place. Such an issue is one of the fundamental problems in financial econometrics and has been widely studied in the literature; see e.g.~\cite{Has1995,Subrahmanyam1997,HMW2002,OWD2017,Has2019}.} 
The exchanges chosen for the current analysis are the NASDAQ and BATS. 
We focus on the component stocks of NASDAQ-100 in 2015, containing totally 108 assets. 
The source is the Daily TAQ database, whose time precision is one micro-second. 
The sample period is the whole month of August 2015, consisting of 21 trading days. 
We use quote data recorded between 9:45 and 15:45. Namely, we discard the first and last 15 minutes from the market opening time to avoid abnormal behavior frequently observed at the start and end of trading sessions. 
To construct log-price processes from quote data, we compute logarithms of micro-prices (cf.~Eq.(2.2) in \cite{GO2010}). 

As argued in \cite{BM2019,Has2019}, the Daily TAQ database provides two kinds of timestamps for each record. The one denotes the time when a quote update is recorded by an exchange matching-engine, while the other denotes the time when it is processed by a securities information processor (SIP). Henceforth we refer to the former as the participant timestamp and to the latter as the SIP timestamp, following \cite{BM2019}. 
In the present analysis we focus mainly on the SIP timestamp because it is assigned by the single SIP managed by the NASDAQ (``NASDAQ SIP'') and thus cause no clock synchronization issue. Nevertheless, we will later consider a geographical effect on lead-lag times. For this purpose we will need to take account of reporting latencies between the NASDAQ SIP and each exchange, requiring us to handle the participant timestamp. Since each exchange is obligated to synchronize their clocks to UTC to within $100\mu\mathrm{s}$, we set our finest time resolution as $\tau_N=100\mu\mathrm{s}$.\footnote{It will also be reasonable to suppose that market participants react at time scales in $100\mu\mathrm{s}$ or slower because it takes at least around $100\mu\mathrm{s}$ to transmit information among the exchanges considered here; see Table 2 in \cite{Tiv2020}.} Then we take
\[
\mathcal{G}^N=\{-10.0\mathrm{ms},-9.9\mathrm{ms},\dots,-0.1\mathrm{ms},0.0\mathrm{ms},0.1\mathrm{ms},\dots,9.9\mathrm{ms},10.0\mathrm{ms}\}
\]
as the search grid. 
We use $L=20$ as the length of Daubechies' wavelet filter. 


For comparison we also compute the following two estimators for lead-lag times. 
\begin{itemize}

\item Hoffmann-Rosenbaum-Yoshida (HRY) estimator \cite{HRY2013}: This estimator is defined as a maximizer of $\widehat{U}^N(\theta)$ over the grid $\mathcal{G}^N$:
\[
\widehat{\theta}^{HRY}=\arg\max_{\theta\in\mathcal{G}^N}|\widehat{U}^N(\theta)|.
\]

\item Dobrev-Schaumburg (DS) estimator \cite{DS2016}: This estimator is constructed as follows. For each $\nu=1,2$ and each $t\geq0$, we set $I^\nu_t=1$ if $t\in\{t^\nu_i:i=0,1,\dots,n_\nu\}$ and $I^\nu_t=0$ otherwise. Then we define 
\[
\mathcal{X}(\theta)=\frac{1}{\min\{n_1,n_2\}}\sum_{k=1}^\infty I^1_{k\tau_N}I^2_{k\tau_N+\theta}
\]
for each $\theta\in\mathbb{R}$. Now, the DS estimator $\widehat{\theta}^{DS}$ is defined as a maximizer of $\mathcal{X}(\theta)$ over the grid $\mathcal{G}^N$:
\[
\widehat{\theta}^{DS}=\arg\max_{\theta\in\mathcal{G}^N}\mathcal{X}(\theta).
\]

\end{itemize}

\subsection{Results}

Figure \ref{third-hist} shows the histograms of the lead-lag time estimates of $\widehat{\theta}^{HRY}$, $\widehat{\theta}_{j}$ ($1\leq j\leq7$) and $\widehat{\theta}^{DS}$ for the 108 assets, evaluated every trading day. 
Here, the horizontal axis is expressed in mili-seconds and the positive values imply that the NASDAQ leads the BATS and vice versa. 
We see that the estimates of $\widehat{\theta}_{j}$ at the levels $j=1,2,3$ have sharp peaks at small positive values, while those at the levels $j=4,5,6,7$ have two peaks located at positive and negative values, respectively. 
The estimates of $\widehat{\theta}^{HRY}$ have a peak around 0 but are negatively skewed. 
The estimates of $\widehat{\theta}^{DS}$ have a very sharp peak at a small positive value, suggesting the presence of consistent leadership of the NASDAQ against the BATS in trading activity. 

These observations suggest that the estimates of $\widehat{\theta}_{j}$ at the finer levels $j=1,2,3$ (corresponding to the time scales between 0.1ms and 0.8ms) might be related to those of $\widehat{\theta}^{DS}$, while the negative estimates of $\widehat{\theta}_{j}$ at the coarser levels $j=4,5,6,7$ (corresponding to the time scales between 0.8ms and 12.8ms) might have some links with those of $\widehat{\theta}^{HRY}$. To confirm the first claim, we compute summary statistics for the estimates of $\widehat{\theta}_{j}$ for $j=1,2,3$ and $\widehat{\theta}^{DS}$ in the left panel of Table \ref{table:fine}. As the table reveals, the estimates of $\widehat{\theta}^{DS}$ are concentrated at $\theta=0.3\mathrm{ms}$ and the other estimates are mostly distributed around this lead-lag time. 
Now, following \cite{DS2016}, we argue that the value $0.3\mathrm{ms}$ comes from a geographical reason. First, we note that our analysis is based on timestamps placed by the NASDAQ SIP and they contain reporting latencies when quote updates are transmitted from each exchange: See \cite{BM2019,Has2019} for more details. 
In particular, Table A.1 in \cite{BM2019} suggests that SIP quote updates for the BATS would primitively lag those for the NASDAQ with lead-lag times around 0.2ms due to the difference between their reporting latencies. 
To check this, we re-evaluate $\widehat{\theta}_{j}$ for $j=1,2,3$ and $\widehat{\theta}^{DS}$ using the participant timestamp instead of the SIP timestamp. The results are reported in the right panel of Table \ref{table:fine}. 
The table shows that the estimates of $\widehat{\theta}^{DS}$ are concentrated at $\theta=0.1\mathrm{ms}$, supporting the above discussion. We speculate that the value $0.1\mathrm{ms}$ is originated from the transit time between the matching engines of the NASDAQ and BATS exchanges: The former's are located in Carteret, NJ, while the latter's are located at the Equinix data center in Secaucus, NJ. The fastest transit time between Carteret and Secaucus is estimated as around 0.1ms; see Table 2 in \cite{Tiv2020}. 

Now we turn to the second claim. Panel A of Table \ref{table:coarse} shows summary statistics for the negative and positive estimates of $\widehat{\theta}_{j}$ for $j=4,5,6,7$ as well as the whole estimates of $\widehat{\theta}^{HRY}$. Here, the row ``Spearman'' in the table indicates Spearman's rank correlation coefficients with the estimates of $\widehat{\theta}^{HRY}$. These values of Spearman's rank correlation coefficients suggest that the negative estimates of $\widehat{\theta}_{j}$ for $j=5,6,7$ would have some relationships with those of $\widehat{\theta}^{HRY}$. 
This becomes more pronounced if we focus on estimates based on larger sample sizes: Panel B of Table \ref{table:coarse} shows the above summary statistics computed for the estimates based on samples with more than 50,000 quote updates. 

In summary, we infer from the above analysis the following lead-lag relationships between the NASDAQ and BATS exchanges for the NASDAQ-100 assets: 
On one hand, the lead-lag time estimates of $\widehat{\theta}^{DS}$ capture cross-market trading activity by the fastest traders and their values come from the geographical distance between the data centers of the two exchanges. At this finest time scales, the NASDAQ exchange typically leads the BATS exchange. 
On the other hand, the lead-lag time estimates of $\widehat{\theta}^{HRY}$ basically capture trading activity by relatively slower traders acting at time scales in mili-seconds. At this relatively coarser time scales, the BATS primally leads the NASDAQ, although there are non-negligible cases that the BATS lags the NASDAQ. 
Our new estimators successfully separate these two types of lead-lag relationships at different time scales up to some extent. 
At the best of our knowledge, this empirical observation is new in the literature, which has been drawn for the first time by our proposed estimation methodology for multiple lead-lag times.

\begin{figure}[h]
\centering
\caption{Histograms of the daily lead-lag time estimates for the NASDAQ-100 assets}
\label{third-hist}
\includegraphics[scale=1]{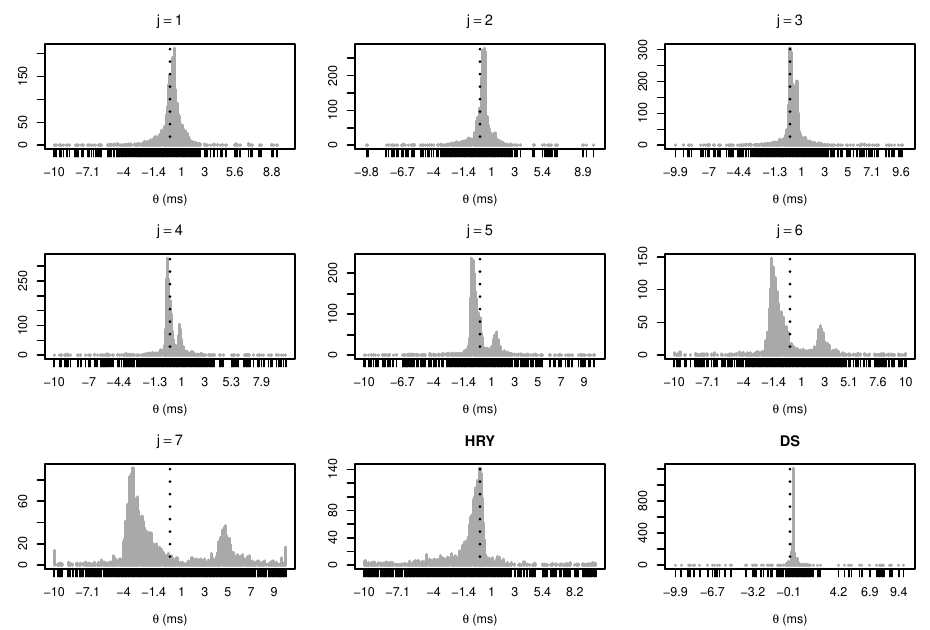}
\medskip

\parbox{15cm}{\small
This figure shows the histograms of the daily lead-lag time estimates $\widehat{\theta}_{j}$ ($1\leq j\leq8$), $\widehat{\theta}^{HRY}$ and $\widehat{\theta}^{DS}$ of quote updates between the NASDAQ and BATS exchanges, computed for all the component stocks of the NASDAQ-100 in August, 2015. 
The price processes are constructed by computing the micro prices. 
The horizontal axis is expressed in mili-seconds. The dash line denotes 0ms. A positive lead-lag time estimate implies that the NASDAQ leads the BATS and vice versa.
}
\end{figure}

\begin{table}[h]
\caption{Summary statistics for the estimates of $\widehat{\theta}_{j}$ for $j=1,2,3$ and $\widehat{\theta}^{DS}$}
\label{table:fine}
\centering
\begin{tabular}{l|rrrr|rrrrr}
  \hline
  & \multicolumn{4}{c|}{SIP timestamp}  & \multicolumn{4}{c}{Participant timestamp} \\
   & $j=1$ & $j=2$ & $j=3$ & DS &   $j=1$ & $j=2$ & $j=3$ & DS \\ \hline
  1st Quartile & -0.20 & 0.00 & 0.00 & 0.20   & -0.10 & -0.20 & -0.20 & -0.10 \\ 
  Median & 0.20 & 0.30 & 0.20 & 0.30 &   0.10 & 0.10 & 0.00 & 0.10 \\ 
  3rd Quartile & 0.60 & 0.50 & 0.60 & 0.30   & 0.20 & 0.30 & 0.40 & 0.10 \\ 
  Mode & 0.40 & 0.40 & 0.00 & 0.30 &   0.10 & 0.30 & -0.10 & 0.10 \\ 
   \hline
\end{tabular}\vspace{2mm}

\parbox{13cm}{\small This table reports the quartiles and modes of the daily lead-lag time estimates $\widehat{\theta}_{j}$ ($j=1,2,3$) and $\widehat{\theta}^{DS}$ (DS) of quote updates between the NASDAQ and BATS exchanges, computed for all the component stocks of the NASDAQ-100 in August, 2015. The values are expressed in mili-seconds. A positive lead-lag time estimate implies that the NASDAQ leads the BATS and vice versa. The left panel is for the estimates based on the SIP timestamp. The right panel is for the estimates based on the participant timestamp.} 
\end{table}

%

\begin{table}[ht]
\caption{Summary statistics for the estimates of $\widehat{\theta}_{j}$ for $j=4,5,6,7$ and $\widehat{\theta}^{HRY}$}
\label{table:coarse}
\centering
\begin{tabular}{lrrrrrrrrrrr}
  \hline
  \multicolumn{12}{l}{Panel A: All the estimates} \\ \hline
& \multicolumn{4}{c}{Negative values} & & \multicolumn{4}{c}{Positive values} \\
    & $j=4$ & $j=5$ & $j=6$ & $j=7$ &  & $j=4$ & $j=5$ & $j=6$ & $j=7$ &  & HRY \\ 
  1st Quartile & -0.40 & -0.70 & -1.60 & -3.50 &  & 0.30 & 0.60 & 1.80 & 4.10 &  & -1.30 \\ 
  Median & -0.30 & -0.60 & -1.30 & -3.00 &  & 0.80 & 1.30 & 2.60 & 4.70 &  & -0.40 \\ 
  3rd Quartile & -0.20 & -0.40 & -0.90 & -2.20 &  & 1.00 & 1.70 & 3.20 & 5.60 &  & 0.10 \\ 
  Mode & -0.20 & -0.70 & -1.60 & -3.20 &  & 0.10 & 1.40 & 2.60 & 4.80 &  & 0.00 \\ 
  Spearman & 0.22 & 0.45 & 0.52 & 0.51 &  & -0.13 & -0.10 & -0.09 & -0.00 &  & 1.00 \\ 
  \# of samples & 1188 & 1543 & 1634 & 1560 &  & 897 & 634 & 618 & 700 &  & 2268 \\ 
   \hline
\hline
 \multicolumn{12}{l}{Panel B: Estimates based on more than 50,000 quote updates in the NASDAQ} \\ \hline
 & \multicolumn{4}{c}{Negative values} & & \multicolumn{4}{c}{Positive values} \\
& $j=4$ & $j=5$ & $j=6$ & $j=7$ &  & $j=4$ & $j=5$ & $j=6$ & $j=7$ &  & HRY \\ 
  1st Quartile & -0.30 & -0.70 & -1.60 & -3.40 &  & 0.10 & 0.20 & 2.40 & 4.40 &  & -1.00 \\ 
  Median & -0.20 & -0.50 & -1.30 & -2.90 &  & 0.40 & 1.30 & 2.70 & 4.80 &  & -0.40 \\ 
  3rd Quartile & -0.10 & -0.30 & -0.80 & -2.10 &  & 0.90 & 1.60 & 3.10 & 5.73 &  & 0.00 \\ 
  Mode & -0.20 & -0.60 & -1.40 & -3.20 &  & 0.10 & 0.10 & 2.50 & 4.80 &  & 0.00 \\ 
  Spearman & 0.44 & 0.69 & 0.73 & 0.75 &  & -0.35 & -0.36 & -0.04 & 0.05 &  & 1.00 \\ 
  \# of samples & 667 & 913 & 990 & 977 &  & 415 & 251 & 220 & 244 &  & 1223 \\ 
   \hline
\end{tabular}\vspace{2mm}

\parbox{15cm}{\small This table reports several summary statistics of the negative and positive estimates of $\widehat{\theta}_{j}$ ($j=4,5,6,7$) as well as the whole estimates of $\widehat{\theta}^{HRY}$ (HRY) for quote updates between the NASDAQ and BATS exchanges, computed for all the component stocks of the NASDAQ-100 in August, 2015. A positive lead-lag time estimate implies that the NASDAQ leads BATS and vice versa. The values for the quartiles and mode are expressed in mili-seconds. The row ``Spearman'' reports Spearman's rank correlation coefficients with the estimates of $\widehat{\theta}^{HRY}$. The row ``\# of estimates'' reports the number of lead-lag time estimates for which the summary statistics are evaluated. The values in Panel A are evaluated with all the estimates. The values in Panel B are based on the estimates for samples with more than 50,000 quote updates in the NASDAQ exchange.} 
\end{table}

\clearpage

\section{Conclusion}\label{section:conclusion}

In this paper we have proposed a new estimation method for multi-scale analysis of lead-lag relationships between two assets based on their high-frequency observation data when they are non-synchronously observed. The key idea was our novel construction of estimators for scale-by-scale cross-covariance functions, where we apply a wavelet transform to the empirical cross-covariance function rather than the raw observation data. We have also developed an associated asymptotic theory to obtain consistency of the proposed estimators in the modeling framework proposed in our previous work \cite{HK2016}. 
Compared with the estimation method proposed in \cite{HK2016}, which essentially adopts the same method as the traditional one in the wavelet literature, the newly proposed method is more appropriate in applications to irregularly-spaced high-frequency financial data as it can handle the whole data more effectively. 
The simulation study has indeed shown that the proposed estimators can perform far better for non-synchronously observed data than the previous one, as intended. The empirical results have demonstrated that the new method can provide a deep insight into lead-lag relationships in the financial markets in the high-frequency domain. 
In particular, we identify two types of lead-lag relationships at finer and (relatively) coarser time scales, respectively.

\section{Proofs}\label{section:proof}


\subsection{Proof of Proposition \ref{prop:bernoulli}}


We begin by proving some lemmas. Let us set $\tilde{\Pi}^\nu_N=\Pi^\nu_N\cup\{(0,t^\nu_0],(t^\nu_{n_1},T+\delta]\}$ for $\nu=1,2$ and $\tilde{r}_N=(\sup_{I\in\tilde{\Pi}^1_N}|I|)\vee(\sup_{J\in\tilde{\Pi}^2_N}|J|)$. We denote by $P^{\Pi^1}$ (resp.~$E^{\Pi^1}$) the conditional probability (resp.~conditional expectation) given $(t^1_i)_{i=0}^{n_1}$.

\begin{lemma}\label{lemma:bernoulli}
Let $\varpi\in(0,1)$ and set $q=\lceil \varpi N\rceil$. Suppose that $\theta_N\geq0$ for all $N$. Then, under the assumptions of Proposition \ref{prop:bernoulli}, there is a constant $C>0$ depending only on $T,\pi_2$ such that
\begin{multline*}
\left|E^{\Pi^1}\left[\frac{1}{2\pi\tau_m\tau_N}\sum_{J\in\tilde{\Pi}^2_N}(\mathcal{F}1_{I})(\lambda/\tau_N)(\mathcal{F}1_{J_{-\theta_N}})(-\lambda/\tau_N)K(I,J_{-\theta_N})\right]\right.\\
\left.-\frac{\tau_N}{\pi\tau_m\lambda^{2}}\Re\left[\left(1-e^{-\sqrt{-1}\lambda\tau_N^{-1}|I|}\right)\frac{1-\pi_2}{1-\pi_2e^{-\sqrt{-1}\lambda/a}}\right]\right|\leq C\tau_m^{-1}\tau_N^{-1}|I|\pi_2^{aT\tau_q\tau_N^{-1}}
\end{multline*}
for any $N\in\mathbb{N}$, any $\lambda\in\mathbb{R}$ and any $I\in\tilde{\Pi}^1_N$ such that $T\tau_q\leq\underline{I}<\overline{I}\leq T(1-\tau_q)$. 
\end{lemma}

\begin{proof}
Fix $\lambda\in\mathbb{R}$ and $I\in\tilde{\Pi}^1_N$ satisfying $T\tau_q\leq\underline{I}<\ol{I}\leq T(1-\tau_q)$ arbitrarily. 
Set
\[
J_I=\bigcup_{J\in\tilde{\Pi}^2_N:K(I,J_{-\theta_N})=1}J_{-\theta_N}.
\]
Then we have
\begin{align*}
&\frac{1}{2\pi\tau_m\tau_N}\sum_{J\in\tilde{\Pi}^2_N}(\mathcal{F}1_{I})(\lambda/\tau_N)(\mathcal{F}1_{J_{-\theta_N}})(-\lambda/\tau_N)K(I,J_{-\theta_N})\\
&=\frac{1}{2\pi\tau_m\tau_N}(\mathcal{F}1_{I})(\lambda/\tau_N)\left\{(\mathcal{F}1_{[\underline{J_I},\underline{I})})(-\lambda/\tau_N)
+(\mathcal{F}1_I)(-\lambda/\tau_N)+(\mathcal{F}1_{[\overline{I},\overline{J_I})})(-\lambda/\tau_N)\right\}\\
&=:\mathbf{I}+\mathbf{II}+\mathbf{III}.
\end{align*}
First we consider $\mathbf{I}$. We can rewrite it as
\[
\mathbf{I}=\frac{\tau_N}{2\pi\tau_m\lambda^{2}}\left(e^{-\sqrt{-1}\lambda\tau_N^{-1}|I|}-1\right)\left(1-e^{-\sqrt{-1}\lambda\tau_N^{-1}(\underline{I}-\underline{J_I})}\right).
\]
Conditionally on $(t^1_i)_{i=0}^{n_1}$, $a\tau_N^{-1}(\underline{I}-\underline{J_I})$ follows the geometric distribution with success probability $1-\pi_2$ truncated from above by $a\tau_N^{-1}\underline{I}$.  More precisely, we have
\[
P^{\Pi^1}(a\tau_N^{-1}(\underline{I}-\underline{J_I})=k)
\left\{
\begin{array}{ll}
\pi_2^k(1-\pi_2)  & \text{if }k=0,1,\dots,a\tau_N^{-1}\underline{I}-1,\\
\pi_2^{a\tau_N^{-1}\underline{I}}  & \text{if }k=a\tau_N^{-1}\underline{I}.    
\end{array}
\right.
\]
Therefore, we obtain
\begin{align*}
&E^{\Pi^1}\left[1-e^{-\sqrt{-1}\lambda\tau_N^{-1}(\underline{I}-\underline{J_I})}\right]\\
&=1-(1-\pi_2)\sum_{k=0}^{a\tau_N^{-1}\underline{I}-1}\pi_2^ke^{-\sqrt{-1}\lambda k/a}-\pi_2^{a\tau_N^{-1}\underline{I}}e^{-\sqrt{-1}\lambda\tau_N^{-1}\underline{I}}\\
&=1-\pi_2^{a\tau_N^{-1}\underline{I}}-\frac{(1-\pi_2)(1-\pi_2^{a\tau_N^{-1}\underline{I}}e^{-\sqrt{-1}\lambda\tau_N^{-1}\underline{I}})}{1-\pi_2e^{-\sqrt{-1}\lambda/a}}+\pi_2^{a\tau_N^{-1}\underline{I}}(1-e^{-\sqrt{-1}\lambda\tau_N^{-1}\underline{I}})\\
&=\frac{\pi_2(1-e^{\sqrt{-1}\lambda/a})}{1-\pi_2e^{-\sqrt{-1}\lambda/a}}-\pi_2^{a\tau_N^{-1}\underline{I}}\frac{1-e^{-\sqrt{-1}\lambda\tau_N^{-1}\underline{I}}-\pi_2(e^{-\sqrt{-1}\lambda/a}-e^{-\sqrt{-1}\lambda\tau_N^{-1}\underline{I}})}{1-\pi_2e^{-\sqrt{-1}\lambda/a}}+\pi_2^{a\tau_N^{-1}\underline{I}}(1-e^{-\sqrt{-1}\lambda\tau_N^{-1}\underline{I}}).
\end{align*}
Consequently, there is a constant $C_1>0$ depending only on $T,\pi_2$ such that
\[
\left|E^{\Pi^1}[\mathbf{I}]-\frac{\tau_N}{2\pi\tau_m\lambda^{2}}\left(e^{-\sqrt{-1}\lambda\tau_N^{-1}|I|}-1\right)\frac{\pi_2(1-e^{-\sqrt{-1}\lambda/a})}{1-\pi_2e^{-\sqrt{-1}\lambda/a}}\right|\leq C_1\tau_m^{-1}\tau_N^{-1}|I|\pi_2^{aT\tau_q\tau_N^{-1}}.
\]
Here, we use the inequality $|1-e^{-\sqrt{-1}x}|\leq|x|$ holding for all $x\in\mathbb{R}$.

Next we consider $\mathbf{III}$. We can rewrite it as
\[
\mathbf{III}=\frac{\tau_N}{2\pi\tau_m\lambda^{2}}\left(1-e^{\sqrt{-1}\lambda\tau_N^{-1}|I|}\right)\left(e^{\sqrt{-1}\lambda\tau_N^{-1}(\overline{J_I}-\overline{I})}-1\right).
\]
Now, an analogous argument to the above yields
\[
\left|E^{\Pi^1}[\mathbf{III}]-\frac{\tau_N}{2\pi\tau_m\lambda^{2}}\left(1-e^{\sqrt{-1}\lambda\tau_N^{-1}|I|}\right)\frac{\pi_2(e^{\sqrt{-1}\lambda/a}-1)}{1-\pi_2e^{\sqrt{-1}\lambda/a}}\right|\leq C_2\tau_m^{-1}\tau_N^{-1}|I|\pi_2^{aT\tau_q\tau_N^{-1}},
\]
where $C_2>0$ is a constant depending only on $T,\pi_2$. Hence we have
\[
\left|E^{\Pi^1}[\mathbf{III}]-\overline{E^{\Pi^1}[\mathbf{I}]}\right|\leq(C_1+C_2)\tau_m^{-1}\tau_N^{-1}|I|\pi_2^{aT\tau_q\tau_N^{-1}}.
\]

Finally, we have 
\[
E^{\Pi^1}[\mathbf{II}]=\mathbf{II}=\frac{\tau_N}{\pi\tau_m\lambda^{2}}\left(1-\Re\left[e^{-\sqrt{-1}\lambda\tau_N^{-1}|I|}\right]\right).
\] 
Therefore, a simple computation yields the desired result.
\end{proof}

\begin{proof}[Proof of Proposition \ref{prop:bernoulli}]
First, similarly to the proof of Eq.(36) form \cite{DY2011}, we can prove
\begin{equation}\label{estimate:r}
E[r_N^p]=O(\tau_N^p|\log\tau_N|^p)
\end{equation}
for any $p>0$. In particular, Assumption \ref{sampling}(i) holds true. 

Next, take a constant $\beta\in(0,1)$ arbitrarily. We prove with this $\beta$ that there are constants $\alpha,Q>1$ such that Assumption \ref{sampling}(ii) holds true. For simplicity of exposition, we assume $\theta_N\geq0$ for all $N$ (this assumption can be easily removed).

Set
\[
\tilde{D}^N_k(\lambda,\theta)=
\frac{1}{2\pi\tau_{m}\tau_N}\sum_{I\in\tilde{\Pi}^1_N,J\in\tilde{\Pi}^2_N:\underline{I}\in I_m(k)}(\mathcal{F}1_{I})(\lambda/\tau_N)(\mathcal{F}1_{J_{-\theta}})(-\lambda/\tau_N)K(I,J_{-\theta}).
\]
One can easily check
\begin{align*}
\tau_{m}\sum_{k=0}^{\lceil T\tau_{m}^{-1}\rceil-1}\int_{-\pi}^\pi E\left[\left|D^N_k(\lambda,\theta_N)-\tilde{D}^N_k(\lambda,\theta_N)\right|^p\right]d\lambda
=O(\tau_{N}^p\tau_m^{-p})
\end{align*}
as $N\to\infty$ for any $p>1$. Therefore, it suffices to show that there are constants $\alpha,Q>1$ such that
\begin{equation}\label{aim:bernoulli}
\tau_{m}\sum_{k=0}^{\lceil T\tau_{m}^{-1}\rceil-1}\int_{-\pi}^\pi E\left[\left|\tilde{D}^N_k(\lambda,\theta_N)-D(\lambda)\right|^Q\right]d\lambda
=O(\tau_{N}^\alpha)
\end{equation} 
as $N\to\infty$. For the proof we adopt an analogous strategy to the proof of Proposition 6 in \cite{DY2011}. Let $\varpi$ be a number such that $\beta<\varpi<1$ and set $q=\lceil N\varpi\rceil$. Let $\mathcal{E}$ be the event on which the interval $I_{q+2}(u)$ contains at least one point from $\{t^1_i:i=0,1,\dots,n_1\}$ and one point from $\{t^2_i:i=0,1,\dots,n_2\}$ for every $u=0,1,\dots,\lfloor T\tau_{q+2}^{-1}\rfloor-1$. We have
\begin{equation}\label{complement}
P(\mathcal{E}^c)\leq T\tau_{q+2}^{-1}(\pi_1^{a\tau_N^{-1}\tau_{q+2}}+\pi_2^{a\tau_N^{-1}\tau_{q+2}}).
\end{equation}
In the following we denote by $E^A$ the conditional expectation given an event $A$. 

For $u\in\mathbb{Z}_+$ and $\lambda\in\mathbb{R}$, we set
\[
\eta_u^N(\lambda)=
\frac{1}{2\pi\tau_m\tau_N}\sum_{I\in\tilde{\Pi}^1_N,J\in\tilde{\Pi}^2_N:\underline{I}\in I_q(u)}(\mathcal{F}1_{I})(\lambda/\tau_N)(\mathcal{F}1_{J_{-\theta}})(-\lambda/\tau_N)K(I,J_{-\theta}).
\]
Then, we decompose $\tilde{D}^N_k(\lambda,\theta_N)-D(\lambda)$ as
\begin{align*}
\tilde{D}^N_k(\lambda,\theta_N)-D(\lambda)
&=\left\{E^\mathcal{E}\left[\sum_{u=k\tau_m\tau_{q}^{-1}}^{(k+1)\tau_m\tau_{q}^{-1}-1}\eta^N_u(\lambda)\right]-D(\lambda)\right\}
\\
&\qquad+\sum_{\begin{subarray}{c}
u=k\tau_m\tau_{q}^{-1}\\
u\text{ is odd}
\end{subarray}}^{(k+1)\tau_m\tau_{q}^{-1}-1}\left(\eta^N_u(\lambda)-E^\mathcal{E}[\eta^N_u(\lambda)]\right)
+\sum_{\begin{subarray}{c}
u=k\tau_m\tau_{q}^{-1}\\
u\text{ is even}
\end{subarray}}^{(k+1)\tau_m\tau_{q}^{-1}-1}\left(\eta^N_u(\lambda)-E^\mathcal{E}[\eta^N_u(\lambda)]\right)\\
&=:\mathbb{I}^N_k(\lambda)+\mathbb{II}^N_k(\lambda)+\mathbb{III}^N_k(\lambda).
\end{align*}

First we consider $\mathbb{I}^N_k(\lambda)$. Using the inequality $|e^{\sqrt{-1}x}-1|\leq|x|$ holding for $x\in\mathbb{R}$, we have 
\begin{equation*}
|\eta^N_u(\lambda)|\leq3(2\pi)^{-1}\tau_m^{-1}\tau_N(\tau_N^{-1}r_N)^2\#\{I\in\tilde{\Pi}^1_N:\underline{I}\in I_q(u)\}.
\end{equation*}
Therefore, noting that $t^1_i-t^1_{i-1}\geq\tau_N/a$ for every $i=1,\dots,n_1-1$, we obtain
\begin{equation}\label{estimate:eta}
|\eta^N_u(\lambda)|\leq3(2\pi)^{-1}a\tau_m^{-1}\tau_q(\tau_N^{-1}r_N)^2.
\end{equation}
Hence we have by \eqref{estimate:r} and \eqref{complement}
\begin{align*}
&\left|E^\mathcal{E}\left[\sum_{u=k\tau_m\tau_{q}^{-1}}^{(k+1)\tau_m\tau_{q}^{-1}-1}\eta^N_u(\lambda)\right]-E\left[\sum_{u=k\tau_m\tau_{q}^{-1}}^{(k+1)\tau_m\tau_{q}^{-1}-1}\eta^N_u(\lambda)\right]\right|\\
&=O\left(\tau_m\tau_q^{-1}\left\{E[\tau_m^{-1}\tau_q(\tau_N^{-1}r_N)^2]P(\mathcal{E}^c)+E[\tau_m^{-1}\tau_q(\tau_N^{-1}r_N)^21_{\mathcal{E}^c}]\right\}\right)\\
&=O\left(\tau_N^2\left\{E[r_N^2]P(\mathcal{E}^c)+\sqrt{E[r_N^4]}\sqrt{P(\mathcal{E}^c)}\right\}\right)\\
&=O\left(|\log\tau_N|^2\tau_{q+2}^{-1}\sqrt{\pi_1^{a\tau_N^{-1}\tau_{q+2}}+\pi_2^{a\tau_N^{-1}\tau_{q+2}}}\right)
\end{align*}
as $N\to\infty$ uniformly in $\lambda$ and $k$. 
Moreover, Lemma \ref{lemma:bernoulli}, \eqref{estimate:r} and \eqref{estimate:eta} imply that
\begin{align*}
&E\left[\sum_{u=k\tau_m\tau_{q}^{-1}}^{(k+1)\tau_m\tau_{q}^{-1}-1}\eta^N_u(\lambda)\right]\\
&=\frac{\tau_N}{\pi\tau_m\lambda^{2}}\Re\left[ E\left[\sum_{I\in\tilde{\Pi}^1_N:\underline{I}\in I_m(k)}\left(1-e^{-\sqrt{-1}\lambda\tau_N^{-1}|I|}\right)\right]\frac{1-\pi_2}{1-\pi_2e^{-\sqrt{-1}\lambda/a}}\right]+O(\tau_N^{\varpi-\beta}|\log\tau_N|^2)
\end{align*}
uniformly in $\lambda$ and $k$. Since $a\tau_N^{-1}(t^1_i-t^1_{i-1})$'s are i.i.d.~variables whose distributions are the geometric distribution with success probability $1-\pi_1$, the Wald identity yields  
\begin{align*}
&E\left[\sum_{u=k\tau_m\tau_{q}^{-1}}^{(k+1)\tau_m\tau_{q}^{-1}-1}\eta^N_u(\lambda)\right]\\
&=\frac{\tau_N}{\pi\tau_m\lambda^{2}}\Re\left[ E\left[\#\{I\in\tilde{\Pi}^1_N:\underline{I}\in I_m(k)\}\right]\frac{(1-\pi_2)(1-e^{-\sqrt{-1}\lambda/a})}{(1-\pi_1e^{-\sqrt{-1}\lambda/a})(1-\pi_2e^{-\sqrt{-1}\lambda/a})}\right]+O(\tau_N^{\varpi-\beta}|\log\tau_N|^2)\\
&=D(\lambda)+O(\tau_N^{\varpi-\beta}|\log\tau_N|^2)
\end{align*}
uniformly in $\lambda$ and $k$. Consequently, we obtain
\[
E\left[|\mathbb{I}^N_k(\lambda)|^p\right]=|\mathbb{I}^N_k(\lambda)|^p=O(\tau_N^{(\varpi-\beta)p}|\log\tau_N|^{2p})
\]
uniformly in $\lambda$ and $k$ for any $p>1$.

Next we consider $\mathbb{II}^n_k(\lambda)$. By construction $(\eta_u^N(\lambda))_{u:\text{ odd}}$ is independent conditionally to $\mathcal{E}$. Therefore, the Burkholder-Davis-Gundy inequality, \eqref{estimate:r} and \eqref{complement}--\eqref{estimate:eta} yield
\begin{align*}
E^\mathcal{E}\left[\left|\mathbb{II}_k^N(\lambda)\right|^p\right]
=O\left((\tau_m\tau_q^{-1})^{p/2}E\left[\left|\tau_m^{-1}\tau_q(\tau_N^{-1}r_N)^2\right|^p\right]\right)
=O\left((\tau_m^{-1}\tau_q)^{p/2}|\log\tau_N|^{2p}\right)
\end{align*} 
uniformly in $\lambda$ and $k$ for any $p>1$. Moreover, \eqref{complement}--\eqref{estimate:r} imply that
$
E^{\mathcal{E}^c}[|\mathbb{II}_k^N(\lambda)|^p]=O((\tau_m^{-1}\tau_q)^{p/2}|\log\tau_N|^{2p})
$ 
uniformly in $\lambda$ and $k$ for any $p>1$. Consequently, we obtain
$
E[|\mathbb{II}_k^N(\lambda)|^p]=O((\tau_m^{-1}\tau_q)^{p/2}|\log\tau_N|^{2p})
$ 
uniformly in $\lambda$ and $k$ for any $p>1$. An analogous argument yields
$
E[|\mathbb{III}_k^N(\lambda)|^p]=O((\tau_m^{-1}\tau_q)^{p/2}|\log\tau_N|^{2p})
$ 
uniformly in $\lambda$ and $k$ for any $p>1$. 

After all, we have
\begin{align*}
\tau_{m}\sum_{k=0}^{\lceil T\tau_{m}^{-1}\rceil-1}\int_{-\pi}^\pi E\left[\left|\tilde{D}^N_k(\lambda,\theta_N)-D(\lambda)\right|^p\right]d\lambda
=O(\tau_N^{(\varpi-\beta)p/2}|\log\tau_N|^{2p})
\end{align*}
for any $p>1$. Now we take $Q>1$ so that $(\varpi-\beta)Q>4$ and set $\alpha=(\varpi-\beta)Q/4$. Then \eqref{aim:bernoulli} holds true.
\end{proof}

\subsection{Proof of Theorem \ref{theorem:main}}\label{proof:theorem1}

First we remark that a standard localization procedure presented e.g.~at the beginning of Section 7.3 of \cite{HK2016} allows us to assume that there is a constant $K>0$ such that
\begin{equation}\label{eq:sigma-bound}
|\sigma^1_t|+|\sigma^2_t|\leq K,\qquad
|\sigma^1_t-\sigma^1_s|+|\sigma^2_t-\sigma^2_s|\leq K|t-s|^\gamma
\end{equation}
for any $t,s\geq0$ throughout the proof.

Next we introduce some notation. For each $k\in\mathbb{Z}_+$ and $\theta\in(-\delta,\delta)$, we set
\[
\widehat{U}^N_k(\theta)=
\left\{\begin{array}{ll}
\sum_{I,J:\underline{I}\in I_m(k)}X^1(I)X^2(J)K(I,J_{-\theta})&\text{if }\theta\geq0,\\
\sum_{I,J:\underline{J}\in I_m(k)}X^1(I)X^2(J)K(I_\theta,J)&\text{if }\theta<0
\end{array}\right.
\]
and
\[
c^N_k(\theta)=
\left\{\begin{array}{ll}
\sigma^1_{k\tau_m}\sigma^2_{(k\tau_m+\theta-r_N)_+}&\text{if }\theta\geq0,\\
\sigma^1_{(k\tau_m-\theta-r_N)_+}\sigma^2_{k\tau_m}&\text{if }\theta<0
\end{array}\right.
\]
and
\[
\widetilde{U}^N_k(\theta)=
\left\{\begin{array}{ll}
\sum_{I,J:\underline{I}\in I_m(k)}B^1(I)B^2(J)K(I,J_{-\theta})&\text{if }\theta\geq0,\\
\sum_{I,J:\underline{J}\in I_m(k)}B^1(I)B^2(J)K(I_\theta,J)&\text{if }\theta<0
\end{array}\right.
\]
and
\[
\overline{U}^N_k(\theta)=
\left\{\begin{array}{ll}
\sum_{I,J:\underline{I}\in I_m(k)}E^\Pi\left[B^1(I)B^2(J)\right]K(I,J_{-\theta})&\text{if }\theta\geq0,\\
\sum_{I,J:\underline{J}\in I_m(k)}E^\Pi\left[B^1(I)B^2(J)\right]K(I_\theta,J)&\text{if }\theta<0,
\end{array}\right.
\]
where $E^\Pi$ denotes the conditional expectation given $(t^1_i)_{i=0}^{n_1}$ and $(t^2_j)_{j=0}^{n_2}$.

For a random variable $Y$ and $p>0$, we set $\|Y\|_{p,\Pi}=(E^\Pi[|Y|^p])^{1/p}$. Also, we denote by $\|Y\|_{\psi_1,\Pi}$ the Orlicz norm of $Y$ based on the function $\psi_1(x)=e^{x}-1$ with respect to the conditional probability given $\Pi$:
\[
\|Y\|_{\psi_1,\Pi}:=\inf\left\{C>0:E^\Pi\left[\psi_1\left(\frac{|Y|}{C}\right)\right]\leq1\right\}.
\]
Throughout this subsection, we write $a\lesssim b$ for two real numbers $a,b$ if $a\leq Cb$ for some constant $C>0$ depending only on $K$.  
\begin{lemma}\label{lemma:block}
Under the assumptions of Theorem \ref{theorem:main}, there is a constant $C>0$ depending only on $K$ such that
\begin{align*}
\left\|\widehat{U}^N_k(\theta)-c^N_k(\theta)\widetilde{U}^N_k(\theta)\right\|_{\psi_1,\Pi}
\leq C(\tau_m+r_N)^{1+\gamma}
\end{align*}
for any $N\in\mathbb{N}$, $k\in\mathbb{Z}_+$ and $\theta\in(-\delta,\delta)$. 
\end{lemma}

\begin{proof}
By symmetry it is enough to consider the case of $\theta\geq0$.

First we apply the so-called reduction procedures used in \cite{HY2008,HY2011} to every realization of $(I)_{I\in\Pi^1_N}$ and $(J_{-\theta})_{J\in\Pi^2_N}$ (see also the proof of Lemma 2 from \cite{DY2011}). We define a new partition $\tilde{\Pi}^1_N$ as follows: $I\in\tilde{\Pi}^1_N$ if and only if either $I\in\Pi^1_N$ and it has non-empty intersection with two distinct intervals from $\Pi^2_N$ or there is $J\in\Pi^2_N$ such that $I$ is the union of all intervals from $\Pi^1_N$ included in $J$. We also define a new partition $\tilde{\Pi}^2_N$ as follows: $J\in\tilde{\Pi}^2_N$ if and only if either $J\in\Pi^2_N$ and $J_{-\theta}$ has non-empty intersection with two distinct intervals from $\Pi^1_N$ or there is $I\in\Pi^1_N$ such that $J$ is the union of all intervals from $J'\in\Pi^2_N$ such that $J'_{-\theta}$ is included in $I$. Due to bilinearity both $\widehat{U}^N_k(\theta)$ and $\widetilde{U}^N_k(\theta)$ are invariant under this procedure. $r_N$ is also unchanged by this application because of its definition. Moreover, by construction we have
\begin{equation*}
\max_{J\in\tilde{\Pi}^2_N}\sum_{I\in\tilde{\Pi}^1_N}K(I,J_{-\theta})\leq3,\qquad
\max_{I\in\tilde{\Pi}^1_N}\sum_{J\in\tilde{\Pi}^2_N}K(I,J_{-\theta})\leq3.
\end{equation*}
Consequently, for the proof we may replace $(\Pi^1_N,\Pi^2_N)$ by $(\tilde{\Pi}^1_N,\tilde{\Pi}^2_N)$. This allows us to assume that
\begin{equation}\label{reduction}
\max_{J\in\Pi^2_N}\sum_{I\in\Pi^1_N}K(I,J_{-\theta})\leq3,\qquad
\max_{I\in\Pi^1_N}\sum_{J\in\Pi^2_N}K(I,J_{-\theta})\leq3.
\end{equation}
throughout the proof without loss of generality.

We turn to the main body of the proof. We decompose the target quantity as
\begin{align*}
&\widehat{U}^N_k(\theta)-c^N_k(\theta)\widetilde{U}^N_k(\theta)\\
&=\sum_{I,J:\underline{I}\in I_m(k)}\left\{\int_I(\sigma^1_s-\sigma^1_{k\tau_m})dB^1_sX^2(J)+\sigma^1_{k\tau_m}B^1(I)\int_J(\sigma^2_s-\sigma^2_{(k\tau_m+\theta-r_N)_+})dB^2_s\right\}K(I,J_{-\theta})\\
&=:\mathbf{A}_N+\mathbf{B}_N.
\end{align*}
Let us consider $\mathbf{A}_N$. For every $p\geq1$, the Minkovski and Schwarz inequalities yield
\begin{align*}
\left\|\mathbf{A}_N\right\|_{p,\Pi}
&\leq\sum_{I,J:\underline{I}\in I_m(k)}\left\|\int_I(\sigma^1_s-\sigma^1_{k\tau_m})dB^1_sX^2(J)\right\|_{p,\Pi}K(I,J_{-\theta})\\
&\leq\sum_{I,J:\underline{I}\in I_m(k)}\left\|\int_I(\sigma^1_s-\sigma^1_{k\tau_m})dB^1_s\right\|_{2p,\Pi}\left\|X^2(J)\right\|_{2p,\Pi}K(I,J_{-\theta}).
\end{align*}
Therefore, Proposition 4.2 in \cite{BY1982} and \eqref{eq:sigma-bound} imply that
\begin{align*}
\left\|\mathbf{A}_N\right\|_{p,\Pi}
&\lesssim p\sum_{I,J:\underline{I}\in I_m(k)}\left\|\sup_{k\tau_m \leq s\leq (k+1)\tau_m+r_N}|\sigma^1_s-\sigma^1_{k\tau_m}|\right\|_{2p,\Pi}\sqrt{|I||J|}K(I,J_{-\theta})\\
&\lesssim p(\tau_m+r_N)^\gamma\sum_{I,J:\underline{I}\in I_m(k)}(|I|+|J|)K(I,J_{-\theta}).
\end{align*}
Thus we obtain $\left\|\mathbf{A}_N\right\|_{p,\Pi}\lesssim p(\tau_m+r_N)^{1+\gamma}$ by \eqref{reduction}. Hence, we conclude $\left\|\mathbf{A}_N\right\|_{\psi_1,\Pi}\lesssim(\tau_m+r_N)^{1+\gamma}$ by Proposition 2.7.1 in \cite{Vershynin2018}. 
By symmetry we also have $\left\|\mathbf{B}_N\right\|_{\psi_1,\Pi}\lesssim(\tau_m+r_N)^{1+\gamma}$. This completes the proof.
\end{proof}

\begin{lemma}\label{lemma:psi-bound}
For any $I\in\Pi^1_N$ and $J\in\Pi^2_N$,
\[
|E^\Pi[B^1(I)B^2(J)]|\leq\sum_{j=1}^{N+1}2^{N-j+1}(|I|\wedge|J|)\int_{\ul{I}-\ol{J}+\theta_j}^{\ol{I}-\ul{J}+\theta_j}|\psi^{LP}(2^{N-j+1}t)|dt.
\]
\end{lemma}

\begin{proof}
Since $B$ has the cross-spectral density $f_N$, we have
\begin{align*}
E^\Pi[B^1(I)B^2(J)]
&=\frac{1}{2\pi}\int_{-\infty}^\infty(\mathcal{F}1_I)(\lambda)(\mathcal{F}1_J)(-\lambda)f_N(\lambda)d\lambda\\
&=\frac{1}{2\pi}\sum_{j=1}^{N+1}R_j\int_{-\infty}^\infty(\mathcal{F}1_I)(\lambda)(\mathcal{F}1_J)(-\lambda)e^{-\sqrt{-1}\theta_j\lambda}1_{\Lambda_{N-j+1}}(\lambda)d\lambda\\
&=\frac{1}{2\pi}\sum_{j=1}^{N+1}R_j\int_{-\infty}^\infty(\mathcal{F}[1_I*1_{-J}])(\lambda)e^{-\sqrt{-1}\theta_j\lambda}1_{\Lambda_0}(\lambda/2^{N-j+1})d\lambda,
\end{align*}
where $*$ denotes convolution. Thus, Parseval's identity yields
\begin{align*}
E^\Pi[B^1(I)B^2(J)]
&=\sum_{j=1}^{N+1}R_j\cdot2^{N-j+1}\int_{-\infty}^\infty(1_I*1_{-J})(t-\theta_j)\psi^{LP}(2^{N-j+1}t)dt\\
&=\sum_{j=1}^{N+1}R_j\cdot2^{N-j+1}\int_{-\infty}^\infty|I\cap(J-\theta_j+t)|\psi^{LP}(2^{N-j+1}t)dt.
\end{align*}
Since $|I\cap(J-\theta_j+t)|\leq|I|\wedge|J|$ and $I\cap(J-\theta_j+t)\neq\emptyset$ only if $\ul{I}-\ol{J}+\theta_j\leq t\leq\ol{I}-\ul{J}+\theta_j$, we obtain the desired result. 
\end{proof}

\begin{lemma}\label{lemma:exp}
Under the assumptions of Theorem \ref{theorem:main}, 
there is a universal constant $C$ such that
\begin{align*}
\left\|\widetilde{U}^N_k(\theta)-\overline{U}^N_k(\theta)\right\|_{\psi_1,\Pi}
\leq C\phi_N
\end{align*}
for any $N\in\mathbb{N}$, $k\in\mathbb{Z}_+$ and $\theta\in(-\delta,\delta)$, where
\[
\phi_N:=\sqrt{r_N(\tau_m+r_N)\left(1+r_N\tau_N^{-1}|\log\tau_N|\right)}.
\]
\end{lemma}

\begin{proof}
Again, by symmetry it suffices to consider the case of $\theta\geq0$. Moreover, as in the proof of Lemma \ref{lemma:block}, we may assume \eqref{reduction} without loss of generality. 

Set $
Z_N:=(B^1(I))_{I\in\Pi^1_N:\underline{I}\in I_m(k)},B^2(J))_{J\in\Pi^2_N:\underline{J}\in \tilde{I}_m(k)})^\top,
$ 
and
\[
A_N=\left(
\begin{array}{cc}
0  &   K_N   \\
K_N^\top  &   0
\end{array}
\right),\qquad
K_N=(K(I,J_{-\theta})/2)_{(I,J)\in\Pi^1_N\times\Pi^2_N:\underline{I}\in I_m(k),\underline{J}\in \tilde{I}_m(k)},
\]
where $\tilde{I}_m(k)=I_m(k)+\theta-r_N$. 
Then we have $\widetilde{U}^N_k(\theta)-\overline{U}^N_k(\theta)=Z_N^\top A_NZ_N-E^\Pi[Z_N^\top A_NZ_N]$. Therefore, by Lemma 13 in \cite{HK2016} there is a universal constant $c_1>0$ such that
\[
\|\widetilde{U}^N_k(\theta)-\overline{U}^N_k(\theta)\|_{\psi_1,\Pi}\leq c_1\sqrt{E^\Pi[|\widetilde{U}^N_k(\theta)-\overline{U}^N_k(\theta)|^2]}.
\]
Let $\Sigma_N$ be the $\Pi$-conditional covariance matrix of $Z_N$ and set $C_{N}=\Sigma_{N}^{1/2}A_N\Sigma_{N}^{1/2}$. Then we have $E^\Pi[|\widetilde{U}^N_k(\theta)-\overline{U}^N_k(\theta)|^2]=2\|C_N\|_F^2$, where $\|\cdot\|_F$ denotes the Frobenius norm of matrices. Hence the proof is completed once we show that 
\begin{equation}\label{matrix}
\|C_{N}\|_F\leq c_2\phi_N^2
\end{equation}
for some universal constant $c_2>0$. 
By Theorem 5.6.9 in \cite{HJ2013} and \eqref{reduction}, we have $\|A_N\|_\mathrm{sp}\leq\frac{3}{2}$. Therefore, inequalities (ii)--(iii) in Appendix II of \cite{Davies1973} yield
\begin{align*}
\|C_{N}\|_F^2&\leq\frac{9}{4}\|\Sigma_{N}\|_F^2=\frac{9}{4}\sum_{I\in\Pi^1_N:\underline{I}\in I_m(k)}E^\Pi\left[B^1(I)^2\right]^2+\frac{9}{4}\sum_{J\in\Pi^2_N:\underline{J}\in \tilde{I}_m(k)}E^\Pi\left[B^2(J)^2\right]^2\\
&\hphantom{\leq9\|\Sigma_{N}\|_F^2=}+\frac{9}{2}\sum_{(I,J)\in\Pi^1_N\times\Pi^2_N:\underline{I}\in I_m(k),\underline{J}\in \tilde{I}_m(k)}E^\Pi\left[B^1(I)B^2(J)\right]^2\\
&\leq\frac{9}{2}\left\{r_N(\tau_m+r_N)+\sum_{(I,J)\in\Pi^1_N\times\Pi^2_N:\underline{I}\in I_m(k),\underline{J}\in \tilde{I}_m(k)}E^\Pi\left[B^1(I)B^2(J)\right]^2\right\}.
\end{align*}
By Lemma \ref{lemma:psi-bound} and using the Schwarz inequality twice, we obtain
\begin{align*}
\left|E^\Pi\left[B^1(I)B^2(J)\right]\right|^2
&\leq(N+1)\sum_{j=1}^{N+1}4^{N-j+1}(|I|\wedge|J|)^2(|I|+|J|)\int_{\ul{I}-\ol{J}+\theta_j}^{\ol{I}-\ul{J}+\theta_j}\psi^{LP}(2^{N-j+1}t)^2dt\\
&\leq2(N+1)r_N^2\sum_{j=1}^{N+1}4^{N-j+1}|I|\int_{\ul{I}-\ol{J}+\theta_j}^{\ol{I}-\ul{J}+\theta_j}\psi^{LP}(2^{N-j+1}t)^2dt.
\end{align*}
Hence we have
\begin{align*}
\sum_{(I,J):\underline{I}\in I_m(k),\underline{J}\in \tilde{I}_m(k)}\left|E^\Pi\left[B^1(I)B^2(J)\right]\right|^2
&\leq2(N+1)r_N^2(\tau_m+r_N)\sum_{j=1}^{N+1}2^{N-j+1}\\
&\leq8(N+1)r_N^2(\tau_m+r_N)\tau_N^{-1}.
\end{align*}
This completes the proof of \eqref{matrix}. 
\end{proof}

\begin{lemma}\label{lemma:ubar}
Under the assumptions of Theorem \ref{theorem:main}, 
we have $\left|\overline{U}^N_k(\theta)\right|\leq6(\tau_m+r_N)$ for any $N\in\mathbb{N}$, $k\in\mathbb{Z}_+$ and $\theta\in(-\delta,\delta)$.
\end{lemma}

\begin{proof}
Similarly to the above proofs, we may assume $\theta\geq0$ and that \eqref{reduction} holds true without loss of generality. Then we have
\begin{align*}
\left|\overline{U}^N_k(\theta)\right|
&=\left|\sum_{I,J:\underline{I}\in I_m(k),\underline{J}\in\tilde{I}_m(k)}E^\Pi\left[B^1(I)B^2(J)\right]K(I,J_{-\theta})\right|\\
&\leq\sum_{I,J:\underline{I}\in I_m(k),\underline{J}\in\tilde{I}_m(k)}\left\{E^\Pi\left[B^1(I)^2\right]+E\left[B^2(J)^2\right]\right\}K(I,J_{-\theta})\\
&\leq3\left\{\sum_{I:\underline{I}\in I_m(k)}|I|+\sum_{I:\underline{J}\in \tilde{I}_m(k)}|J|\right\}
\leq3\cdot2(\tau_m+r_N).
\end{align*}
This completes the proof.
\end{proof}

\begin{lemma}\label{lemma:main}
Under the assumptions of Theorem \ref{theorem:main}, 
we have
\begin{align*}
\max_{\theta\in\mathcal{G}^N}\left|\widehat{\rho}_{\tcr{N-j+1}}(\theta)-\tau_m\sum_{k=0}^{\lfloor T\tau_m^{-1}\rfloor-1}c^N_k(\theta)\int_{-\pi}^\pi D(\lambda)H_{j,L}(\lambda)e^{\sqrt{-1}\lambda\theta\tau_N^{-1}}f_N(\lambda/\tau_N)d\lambda\right|\to^p0
\end{align*}
as $N\to\infty$.
\end{lemma}

\begin{proof}
We decompose the target quantity as
\begin{align*}
&\widehat{\rho}_{\tcr{N-j+1}}(\theta)-\tau_m\sum_{k=0}^{\lfloor T\tau_m^{-1}\rfloor-1}c^N_k(\theta)\int_{-\pi}^\pi D(\lambda)H_{j,L}(\lambda)e^{\sqrt{-1}\lambda\theta\tau_N^{-1}}f_N(\lambda/\tau_N)d\lambda\\
&=\left(\widehat{\rho}_{\tcr{N-j+1}}(\theta)-\sum_{l=-L_j-1}^{L_j-1}\Psi_j(l)\sum_{k=0}^{\lfloor T\tau_m^{-1}\rfloor-1}c^N_k(\theta-l\tau_N)\widetilde{U}^N_k(\theta-l\tau_N)\right)\\
&\qquad+\sum_{l=-L_j+1}^{L_j-1}\Psi_j(l)\sum_{k=0}^{\lfloor T\tau_m^{-1}\rfloor-1}c^N_k(\theta-l\tau_N)\left\{\widetilde{U}^N_k(\theta-l\tau_N)-\overline{U}^N_k(\theta-l\tau_N)\right\}\\
&\qquad+\sum_{l=-L_j+1}^{L_j-1}\Psi_j(l)\sum_{k=0}^{\lfloor T\tau_m^{-1}\rfloor-1}\left\{c^N_k(\theta-l\tau_N)-c^N_k(\theta)\right\}\overline{U}^N_k(\theta-l\tau_N)\\
&\qquad+\sum_{k=0}^{\lfloor T\tau_m^{-1}\rfloor-1}c^N_k(\theta)\left(\sum_{l=-L_j+1}^{L_j-1}\Psi_j(l)\overline{U}^N_k(\theta-l\tau_N)-\tau_m\int_{-\pi}^\pi D(\lambda)H_{j,L}(\lambda)e^{\sqrt{-1}\lambda\theta\tau_N^{-1}}f_N(\lambda/\tau_N)d\lambda\right)\\
&=:\mathbf{I}_N(\theta)+\mathbf{II}_N(\theta)+\mathbf{III}_N(\theta)+\mathbf{IV}_N(\theta).
\end{align*}
We have
\begin{align*}
\left|\mathbf{I}_N(\theta)\right|\leq
\sum_{l=-L_j-1}^{L_j-1}|\Psi_j(l)|\left\{\sum_{k=0}^{\lfloor T\tau_m^{-1}\rfloor-1}\left|\widehat{U}^N_k(\theta-l\tau_N)-c^N_k(\theta-l\tau_N)\widetilde{U}^N_k(\theta-l\tau_N)\right|+\left|\widehat{U}^N_{\lfloor T\tau_m^{-1}\rfloor}(\theta-l\tau_N)\right|\right\}.
\end{align*}
Hence, we obtain by Lemmas \ref{lemma:block} and \ref{lemma:exp}--\ref{lemma:ubar}
\begin{align*}
\left\|\mathbf{I}_N(\theta)\right\|_{\psi_1,\Pi}
&\lesssim \{(\tau_m+r_N)^{\gamma}+\phi_N\}\sum_{l=-L_j-1}^{L_j-1}|\Psi_j(l)|.
\end{align*}
So Lemma 2.2.2 in \cite{VW1996} yields
\[
\left\|\max_{\theta\in\mathcal{G}_N}\left|\mathbf{I}_N(\theta)\right|\right\|_{\psi_1,\Pi}
\lesssim |\log\tau_N|\{(\tau_m+r_N)^{\gamma}+\phi_N\}\sum_{l=-L_j-1}^{L_j-1}|\Psi_j(l)|.
\] 
Now, we have by \eqref{acw-bound} and Corollary \ref{coro:decay}
\begin{equation}\label{eq:acw-l1}
\sum_{l=-L_j-1}^{L_j-1}|\Psi_j(l)|=O(\log L)
\end{equation}
as $L\to\infty$. Hence we conclude $\max_{\theta\in\mathcal{G}^N}\left|\mathbf{I}_N(\theta)\right|\to^p0$ by assumption.

Next, Lemma \ref{lemma:exp} yields
\[
\max_{\theta\in\mathcal{G}_N}\left\|\mathbf{II}_N(\theta)\right\|_{\psi_1,\Pi}
\lesssim \tau_m^{-1}\phi_N\sum_{l=-L_j-1}^{L_j-1}|\Psi_j(l)|,
\]
so we obtain by Lemma 2.2.2 in \cite{VW1996} 
\[
\left\|\max_{\theta\in\mathcal{G}_N}\left|\mathbf{II}_N(\theta)\right|\right\|_{\psi_1,\Pi}
\lesssim |\log\tau_N|\tau_m^{-1}\phi_N\sum_{l=-L_j-1}^{L_j-1}|\Psi_j(l)|.
\]
Since $\tau_m^{-1}\phi_N=O_p(\tau_N^a)$ for some $a>0$ and $L=O(\tau_N^{-1})$ by assumption, we conclude $\max_{\theta\in\mathcal{G}^N}\left|\mathbf{II}_N(\theta)\right|\to^p0$ from \eqref{eq:acw-l1}.

Now we prove $\max_{\theta\in\mathcal{G}^N}\left|\mathbf{III}_N(\theta)\right|\to^p0$. We have
\begin{align*}
&\max_{\theta\in\mathcal{G}^N}\left|\mathbf{III}_N(\theta)\right|\\
&\leq\max_{\theta\in\mathcal{G}^N}\max_{l\in\mathbb{Z}:|l|< L_j}\max_{k=0,1,\dots,\lfloor T\tau_m^{-1}\rfloor-1}\left|c^N_k(\theta-l\tau_N)-c^N_k(\theta)\right|\left(\sum_{l=-L_j-1}^{L_j-1}|\Psi_j(l)|\right)\left(\sum_{k=0}^{\lfloor T\tau_m^{-1}\rfloor-1}\max_{\theta\in\mathcal{G}^N}\left|\overline{U}^N_k(\theta)\right|\right).
\end{align*}
By the H\"older continuity of $\sigma^1,\sigma^2$ and \eqref{eq:acw-l1}, we have
\[
\max_{\theta\in\mathcal{G}^N}\max_{l\in\mathbb{Z}:|l|< L_j}\max_{k=0,1,\dots,\lfloor T\tau_m^{-1}\rfloor-1}\left|c^N_k(\theta-l\tau_N)-c^N_k(\theta)\right|\left(\sum_{l=-L_j-1}^{L_j-1}|\Psi_j(l)|\right)=O_p((L\tau_N)^\gamma\log L)
\]
as $N\to\infty$, while Lemma \ref{lemma:ubar} yields $\sum_{k=0}^{\lfloor T\tau_m^{-1}\rfloor-1}\max_{\theta\in\mathcal{G}^N}\left|\overline{U}^N_k(\theta)\right|=O_p(1)$. Hence we obtain the desired result by assumption.

Finally we prove $\max_{\theta\in\mathcal{G}^N}\left|\mathbf{IV}_N(\theta)\right|\to^p0$. Noting that
\begin{align*}
\int_{-\pi}^\pi D(\lambda)H_{j,L}(\lambda)e^{\sqrt{-1}\lambda\theta\tau_N^{-1}}f_N(\lambda/\tau_N)d\lambda
=\sum_{l=-L_j-1}^{L_j-1}\Psi_j(l)\int_{-\pi}^\pi D(\lambda)e^{\sqrt{-1}\lambda(\theta-l\tau_N)\tau_N^{-1}}f_N(\lambda/\tau_N)d\lambda,
\end{align*}
we have for any $\varepsilon>0$
\begin{align*}
&P\left(\max_{\theta\in\mathcal{G}^N}\left|\mathbf{IV}_N(\theta)\right|>\varepsilon\right)\\
&\leq\sum_{\theta\in\mathcal{G}^N}P\left(\tau_m\sum_{l=-L_j+1}^{L_j-1}|\Psi_j(l)|\sum_{k=0}^{\lfloor T\tau_m^{-1}\rfloor-1}\left|\tau_m^{-1}\overline{U}^N_k(\theta-l\tau_N)-\int_{-\pi}^\pi D(\lambda)e^{\sqrt{-1}\lambda(\theta-l\tau_N)\tau_N^{-1}}f_N(\lambda/\tau_N)d\lambda\right|>\frac{\varepsilon}{K^2}\right)
\end{align*}
by \eqref{eq:sigma-bound}. Since we have
\begin{align*}
\overline{U}^N_k(\theta-l\tau_N)
&=\tau_m\int_{-\pi}^\pi D^N_k(\lambda,\theta-l\tau_N)e^{\sqrt{-1}\lambda(\theta-l\tau_N)\tau_N^{-1}}f_N(\lambda/\tau_N)d\lambda,
\end{align*}
it holds that
\begin{align*}
&E\left[\left\{\tau_m\sum_{l=-L_j+1}^{L_j-1}|\Psi_j(l)|\sum_{k=0}^{\lfloor T\tau_m^{-1}\rfloor-1}\left|\tau_m^{-1}\overline{U}^N_k(\theta-l\tau_N)-\int_{-\pi}^\pi D(\lambda)e^{\sqrt{-1}\lambda(\theta-l\tau_N)}f_N(\lambda/\tau_N)d\lambda\right|\right\}^Q\right]\\
&\leq\left(2\pi T\sum_{l=-L_j+1}^{L_j-1}|\Psi_j(l)|\right)^{Q-1}\tau_m\sum_{k=0}^{\lfloor T\tau_m^{-1}\rfloor-1}\int_{-\pi}^\pi E\left[\left|D^N_k(\lambda,\theta-l\tau_N)-D(\lambda)\right|^Q\right]d\lambda
\end{align*}
by the Jensen inequality. Therefore, by the Markov inequality we obtain
\begin{align*}
&P\left(\max_{\theta\in\mathcal{G}^N}\left|\mathbf{IV}_N(\theta)\right|>\varepsilon\right)\\
&\leq\left(\frac{K^2}{\varepsilon}\right)^Q\tau_N^{-1}\left(2\pi T\sum_{l=-L_j+1}^{L_j-1}|\Psi_j(l)|\right)^{Q-1}\max_{\theta\in\mathcal{G}^N}\tau_m\sum_{k=0}^{\lfloor T\tau_m^{-1}\rfloor-1}\int_{-\pi}^\pi E\left[\left|D^N_k(\lambda,\theta)-D(\lambda)\right|^Q\right]d\lambda.
\end{align*}
Consequently, Assumption \ref{sampling} and \eqref{eq:acw-l1} imply the desired result (note that $L=O(\tau_N^{-1})$ by assumption). This completes the proof.
\end{proof}


\begin{proof}[\upshape{\textbf{Proof of Theorem \ref{theorem:main}}}]

(a) From Lemma \ref{lemma:main} it is enough to prove
\[
\max_{\theta\in\mathcal{G}^N:|\theta-\theta_j|\geq v_N}\left|\tau_m\sum_{k=0}^{\lfloor T\tau_m^{-1}\rfloor-1}c^N_k(\theta)\int_{-\pi}^\pi D(\lambda)H_{j,L}(\lambda)e^{\sqrt{-1}\lambda\theta\tau_N^{-1}}f_N(\lambda/\tau_N)d\lambda\right|\to^p0
\]
as $N\to\infty$. The above equation follows once we show the following statements: If $\vartheta_N\in\mathcal{G}^N$ ($N=1,2,\dots$) satisfy $|\vartheta_N-\theta_j|\geq v_N$ for every $N$, then
\[
a_N:=\int_{-\pi}^\pi D(\lambda)H_{j,L}(\lambda)e^{\sqrt{-1}\lambda\vartheta_N\tau_N^{-1}}f_N(\lambda/\tau_N)d\lambda\to0
\]
as $N\to\infty$. To prove this statement, we decompose $a_N$ as
\[
a_N=\sum_{i=1}^{N+1}R_i\int_{\Lambda_{-i}} D(\lambda)H_{j,L}(\lambda)e^{\sqrt{-1}\lambda(\vartheta_N-\theta_i)\tau_N^{-1}}d\lambda
=:\sum_{i=1}^{N+1}a_N(i).
\]
Since $|a_N(i)|\leq 2^j\|D\|_\infty\pi/2^{i-1}$ by \eqref{eq:daub-bound} ($\|D\|_\infty$ denotes the essential supremum of $D$), it suffices to prove $a_N(i)\to0$ as $N\to\infty$ for any fixed $i$ due to the dominated convergence theorem. When $i\neq j$, this follows from \eqref{lai} and the bounded convergence theorem. In the meantime, using integration by parts, we obtain
\begin{align*}
|a_N(j)|&\leq\frac{2^{j+1}\|D\|_\infty}{\tau_N^{-1}|\vartheta_N-\theta_j|}
+\frac{1}{\tau_N^{-1}|\vartheta_N-\theta_j|}\int_{\Lambda_{-j}} (|D'(\lambda)H_{j,L}(\lambda)|+|D(\lambda)H_{j,L}'(\lambda)|)d\lambda\\
&\leq\frac{2^{j+1}\|D\|_\infty}{\tau_N^{-1}|\vartheta_N-\theta_j|}
+\frac{1}{\tau_N^{-1}|\vartheta_N-\theta_j|}\left(
2\pi\|D'\|_\infty
+\|D\|_\infty\int_{\Lambda_{-j}}|H_{j,L}'(\lambda)|d\lambda
\right).
\end{align*}
Hence we obtain $a_N(j)\to0$ as $N\to\infty$ by Corollary \ref{coro:deriv} and assumption. 


(b) From Lemma \ref{lemma:main} it suffices to prove
\[
\tau_m\sum_{k=0}^{\lfloor T\tau_m^{-1}\rfloor-1}c^N_k(\theta)\int_{-\pi}^\pi D(\lambda)H_{j,L}(\lambda)e^{\sqrt{-1}\lambda\theta\tau_N^{-1}}f_N(\lambda/\tau_N)d\lambda
\to^p2^j\Sigma_T(\theta_j)R_{j}\int_{\Lambda_{-j}}D(\lambda)\cos(b\lambda)d\lambda
\]
as $N\to\infty$. Using an analogous argument to the above, we can deduce this convergence from the dominated convergence theorem and \eqref{lai}. 
\end{proof}

\subsection{Proof of Theorem \ref{HRY}}

Noting that $\int_{\Lambda_{-j}}D(\lambda)\cos(b\lambda)d\lambda>0$ for any $b\in[-\frac{1}{2},\frac{1}{2}]$ by assumption, the theorem can be shown in an analogous manner to the proof of Theorem 2 in \cite{HK2016} (using Theorem 1 instead of Theorem 1(b) and Proposition 3 in \cite{HK2016}).\hfill$\Box$

\appendix

\section{Appendix: Fundamental properties of Daubechies' wavelet filter}\label{sec:appendix}

\setcounter{equation}{0}
\numberwithin{equation}{section}

This appendix collects a few key properties of Daubechies' wavelet filter used in the proofs of our main results. 
First, since $H_L(\lambda)+G_L(\lambda)=2$ for all $\lambda\in\mathbb{R}$ by Eq.(69d) in \cite{PW2000}, we have
\begin{equation}\label{eq:daub-bound}
H_L(\lambda)\vee G_L(\lambda)\leq2\qquad\text{for all }\lambda\in\mathbb{R}.
\end{equation}
Next, for every $j\in\mathbb{N}$, the filter $(h_{j,p})_{p=0}^{L_j-1}$ has unit energy: $\sum_{p=0}^{L_j-1}h_{j,p}^2=1$ (cf.~Section 4.6 of \cite{PW2000}). Hence, the Schwarz inequality yields
\begin{equation}\label{acw-bound}
|\Psi_j(l)|\leq1\qquad\text{for all }l=0,\pm1,\dots,\pm(L_j-1). 
\end{equation}
Third, $H_{j,L}$ well approximates $2^j1_{\Lambda_{-j}}$ as $L\to\infty$ in the sense that
\begin{equation}\label{lai}
\lim_{L\to\infty}H_{j,L}(\lambda)
=\left\{
\begin{array}{ll}
2^j  & \text{if }\lambda\in(\frac{\pi}{2^j},\frac{\pi}{2^{j-1}}),\\
0  & \text{if }\lambda\in[0,\frac{\pi}{2^j})\cup(\frac{\pi}{2^{j-1}},\pi]
\end{array}
\right.
\end{equation}
by Theorem 1 in \cite{Lai1995}. Here, note that \citet{Lai1995} defines Daubechies' wavelet filter of length $L$ as a filter whose power transfer function is given by $H_L(\lambda)/2$.

Finally, we prove an important property of the derivatives of $H_L$ and $G_L$ as well as its consequences.  
\begin{lemA}\label{lem:daub-deriv}
For any even positive integer $L$, 
\begin{equation}\label{eq:daub-deriv}
\int_{-\pi}^\pi|H_L'(\lambda)|d\lambda=\int_{-\pi}^\pi|G_L'(\lambda)|d\lambda=4.
\end{equation}
\end{lemA}

\begin{proof}
From the proof of Theorem 3 in \cite{Lai1995} we have
\[
G_L'(\lambda)=-\binom{L-2}{L/2-1}\frac{L-1}{2^{L-2}}\sin^{L-1}(\lambda),\qquad\lambda\in[0,\pi].
\]
Thus we obtain
\begin{align*}
\int_{-\pi}^\pi|G_L'(\lambda)|d\lambda
&=2\binom{L-2}{L/2-1}\frac{L-1}{2^{L-2}}\int_0^\pi\sin^{L-1}(\lambda)d\lambda\\
&=2\binom{L-2}{L/2-1}\frac{L-1}{2^{L-2}}\int_0^{\pi/2}\{\sin^{L-1}(\lambda)+\cos^{L-1}(\lambda)\}d\lambda.
\end{align*}
Using Eq.(3.621.1) in \cite{GR2007}, we infer that
\begin{align*}
\int_{-\pi}^\pi|G_L'(\lambda)|d\lambda
&=4\binom{L-2}{L/2-1}(L-1)\frac{(L/2-1)!(L/2-1)!}{(L-1)!}=4.
\end{align*}
So we obtain the second identity in \eqref{eq:daub-deriv}. The first identity in \eqref{eq:daub-deriv} immediately follows from the relation $H_L(\lambda)=G_L(\lambda+\pi)$. 
\end{proof}

\begin{coroA}\label{coro:deriv}
For any $j\in\mathbb{N}$ and even positive integer $L$, 
\[
\int_{\Lambda_{-j}}|H_{j,L}'(\lambda)|d\lambda\leq2^{j+1}j.
\]
\end{coroA}

\begin{proof}
Using the Leibniz rule, we deduce 
\begin{equation}\label{eq:leibniz}
|H_{j,L}'(\lambda)|
\leq4^{j-1}|H_L'(2^{j-1}\lambda)|+2^{j-1}\sum_{i=0}^{j-2}2^i|G_L'(2^i\lambda)|.
\end{equation}
So we obtain by \eqref{eq:daub-bound} and Lemma \ref{lem:daub-deriv}
\begin{align*}
\int_{\Lambda_{-j}}|H_{j,L}'(\lambda)|d\lambda
\leq2^{j-1}\int_{-\pi}^\pi|H_L'(\lambda)|d\lambda+2^{j-1}\sum_{i=0}^{j-2}\int_{-\pi}^\pi|G_L'(\lambda)|d\lambda
\leq2^{j+1}j.
\end{align*}
This completes the proof.
\end{proof}

\begin{coroA}\label{coro:decay}
For any $j\in\mathbb{N}$, even positive integer $L$ and $l=\pm1,\dots,\pm(L_j-1)$, 
\[
|\Psi_j(l)|\leq\frac{4^j}{\pi l}.
\]
\end{coroA}

\begin{proof}
By definition we have
\[
\Psi_j(l)=\frac{1}{2\pi}\int_{-\pi}^\pi H_{j,L}(\lambda)e^{\sqrt{-1}l\lambda}d\lambda.
\]
Hence, integration by parts yields
\[
\Psi_j(l)=-\frac{1}{2\pi\sqrt{-1}l}\int_{-\pi}^\pi H'_{j,L}(\lambda)e^{\sqrt{-1}l\lambda}d\lambda.
\]
Therefore, noting that both $|H_L|$ and $|G_L|$ are periodic with period $\pi$, we obtain by \eqref{eq:leibniz}
\begin{align*}
|\Psi_j(l)|&\leq\frac{1}{2\pi l}\left\{4^{j-1}\int_{-\pi}^\pi |H_L'(2^{j-1}\lambda)|d\lambda+2^{j-1}\sum_{i=0}^{j-2}2^i\int_{-\pi}^\pi|G_L'(2^i\lambda)|d\lambda\right\}\\
&=\frac{1}{2\pi l}\left\{2^{j-1}\int_{-2^{j-1}\pi}^{2^{j-1}\pi} |H_L'(\lambda)|d\lambda+2^{j-1}\sum_{i=0}^{j-2}\int_{-2^i\pi}^{2^i\pi}|G_L'(\lambda)|d\lambda\right\}\\
&=\frac{1}{2\pi l}\left\{4^{j-1}\int_{-\pi}^{\pi} |H_L'(\lambda)|d\lambda+2^{j-1}\sum_{i=0}^{j-2}2^i\int_{-\pi}^{\pi}|G_L'(\lambda)|d\lambda\right\}
\leq\frac{4^j}{\pi l},
\end{align*}
where we used Lemma \ref{lem:daub-deriv} and the inequality $\sum_{i=0}^{j-2}2^i\leq2^{j-1}$ to deduce the last inequality.
\end{proof}

\section*{Acknowledgments}

Takaki Hayashi's research was partly supported by JSPS KAKENHI Grant Numbers JP16K03601, JP17H01100. 
Yuta Koike's research was partly supported by JST CREST Grant Number JPMJCR14D7 and JSPS KAKENHI Grant Number JP16K17105, JP18H00836, JP19K13668. 

{\small
\renewcommand*{\baselinestretch}{1}\selectfont
\addcontentsline{toc}{section}{References}

}

\end{document}